\documentclass[runningheads,a4paper]{llncs}

\usepackage{amssymb,latexsym,amsmath}     

\usepackage{amsfonts}
\usepackage{graphicx}
\usepackage{tikz}
\usepackage{textcomp}
\usepackage{times}
\usepackage[multiple]{footmisc}
\usepackage{mathtools}
\usepackage{tikz}
\usetikzlibrary{decorations.pathreplacing}
\usepackage[center]{subfigure}
\usepackage[numbers]{natbib}
\usetikzlibrary{matrix,decorations.pathreplacing}

\usepackage{algorithm}
\usepackage{algorithmicx}
\usepackage{algpseudocode}

\newenvironment{proofof}[1]{\begin{proof}[Proof of Theorem {#1}]}{\end{proof}}
\usepackage{myDefs}
\usepackage{url}

\begin{document}
\mainmatter
\title{Network Cournot Competition
}

\author{Melika Abolhassani\inst{1} \and MohammadHossein Bateni\inst{2} \and
MohammadTaghi Hajiaghayi \inst{1}\and Hamid Mahini \inst{1} \and Anshul Sawant \inst{1}}

\institute{University of Maryland, College Park, USA,\\
\and
Google Research, New York, USA}

\titlerunning{Network Cournot Competition}
\authorrunning{Abolhassani et al.}

\maketitle
\begin{abstract}
Cournot competition, introduced in 1838 by Antoine Augustin Cournot, 
is a fundamental economic model that represents firms competing in a single market of a homogeneous good.
Each firm tries to maximize its utility---naturally a function of the production cost as well as market price of the product---by deciding
on the amount of production.
This problem has been studied comprehensively in Economics and Game Theory; 
however, in today's dynamic and diverse economy, many firms often compete 
in more than one market simultaneously, i.e., each market might be shared 
among a subset of these firms. 
In this situation, a bipartite graph models the access restriction 
where firms are on one side, markets are on the other side, and edges 
demonstrate whether a firm has access to a market or not.
We call this game \emph{Network Cournot Competition} (NCC). 
Computation of equilibrium, taking into account a network of markets and 
firms and the different forms of cost and price functions, makes 
challenging and interesting new problems.
  
In this paper, we propose algorithms for finding pure Nash equilibria of NCC games 
in different situations.
First, we carefully design a potential function for NCC, when the price functions 
for markets are linear functions of the production in that market. This result 
lets us leverage optimization techniques for a single function rather than 
multiple utility functions of many firms. 
However, for nonlinear price functions, this approach is not feasible---there is
indeed no single potential function that captures the utilities of all firms for the case of nonlinear price functions. We model 
the problem as a nonlinear complementarity problem in this case, and design a 
polynomial-time algorithm that finds an equilibrium of the game for strongly 
convex cost functions and strongly monotone revenue functions. We also explore 
the class of price functions that ensures strong monotonicity of the revenue 
function, and show it consists of a broad class of functions. Moreover, we discuss the uniqueness of equilibria in both of these cases which means our algorithms find the unique equilibria of the games.
Last but not least, when the cost of production in one market is independent from the 
cost of production in other markets for all firms, the problem can be separated 
into several independent classical \emph{Cournot Oligopoly} problems in which 
the firms compete over a single market. We give the first combinatorial algorithm 
for this widely studied problem. Interestingly, our algorithm is much simpler 
and faster than previous optimization-based approaches.
\end{abstract}

\section{Introduction}

In this paper we study selling a utility with a distribution network, e.g., natural gas, water and electricity, in several  markets when the clearing price of each market is determined by its supply and demand. The distribution network fragments the market into different regional markets with their own prices. Therefore, the relations between suppliers and submarkets form a complex network~\cite{neuhoff2005network, jing1999spatial, pow02, pow05, D05, gas05, EU06}. For example, a market with access to only one supplier suffers a monopolistic price, while a market having access to multiple suppliers enjoys a lower price as a result of the price competition.

Antoine Augustin Cournot introduced the first model for studying the 
duopoly competition in 1838. He proposed a model where two individuals own different 
springs of water, and sell it independently. Each individual decides on the amount of water 
to supply, and then the aggregate water supply determines the market price through 
an inverse demand function. Cournot characterizes the unique equilibrium outcome of the market 
when both suppliers have the same marginal costs of production, and the inverse demand function 
is linear. He argued that in the unique equilibrium outcome, the market price is above the 
marginal cost.

Joseph  Bertrand~\citeyear{Bertrand1883} criticized the Cournot model,
where the strategy of each player is the quantity to supply, 
and in turn suggested to consider prices, rather than quantities, as strategies. 
In the Bertrand model each firm chooses a price for a homogeneous good, 
and the firm announcing the lowest price gets all the market share. 
Since the firm with the lowest price receives all the demand, each firm has incentive 
to price below the current market price unless the market price matches its cost. 
Therefore, in an equilibrium outcome of the Bertrand model, assuming all marginal costs 
are the same and there are at least two competitors in the market, the market price 
will be equal to the marginal cost.

The Cournot and Bertrand models are two basic tools for investigating the competitive market price, 
and have attracted much interest for modeling real markets; see, e.g., \cite{pow02,pow05,D05,gas05}. 
While these are two extreme models for analyzing the price competition, 
it is hard to say which one is essentially better than the other. 
In particular, the predictive power of each strongly depends on the nature of the market, 
and varies from application to application. 
For example, the Bertrand model explains the situation where firms literally set prices, 
e.g., the cellphone market,  the laptop market, and   the TV market.
On the other hand, Cournot's approach would be suitable for modeling markets
like those of crude oil, natural gas, and  electricity, where firms 
decide about quantities rather than prices.

There are several attempts to find equilibrium outcomes of the Cournot or Bertrand competitions 
in the oligopolistic setting, where a small number of firms compete in only one market; 
see, e.g., \cite{KS83,SV84,OP86,H90,H00,V01}. 
Nevertheless, it is not entirely clear what equilibrium outcomes of these games are 
when firms compete over more than one market.
%
%
In this paper, we investigate the problem of finding equilibrium outcomes 
of the Cournot competition in a network setting where there are several markets 
for a homogeneous good and each market is accessible to a subset of firms. 
%
\subsection{Example}
We start with the following warm-up example.
This is a basic example for the Cournot competition in the network setting. 
It consists of three scenarios. We assume firm $i \in \{A,B\}$ 
produces quantity $q_{ij}$ of the good in market $j\in\{1,2\}$. 
Let $\mathbf{q}$ be the vector of all quantities. The detailed computations are in the Appendix.
\begin{description}
\item[Scenario 1]
Consider the Cournot competition in an oligopolistic setting with two firms and one market 
(see Figure~\ref{fig1}). Let $p(\mathbf{q}) = 1 - q_{A1}-q_{B1}$ be the market price 
(the inverse demand function), and $c_i(\mathbf{q})= \frac{1}{2}q_{i1}^2$ be the cost of production 
for firm $i \in \{A,B\}$. 
The profit of a firm is what it gets by selling all the quantities of good it produces 
in all markets minus its cost of production.
Therefore, the profit of firm $i$ denoted by $\pi_i(\mathbf{q})$ is 
$q_{i1} (1-q_{A1}-q_{B1}) - \frac{1}{2}q_{i1}^2$. In a Nash equilibrium of the game, 
each firm maximizes its profit assuming its opponent does not change its strategy. 
Hence, the unique Nash equilibrium of the game can be found by solving the set of 
equations $\frac{\partial \pi_A}{\partial q_{A1}}=\frac{\partial \pi_B}{\partial q_{B1}}=0$. So $q_{A1}=q_{B1}=\frac{1}{4}$ is the unique Nash equilibrium 
where $p(\mathbf{q})= \frac{1}{2}$, and $\pi_A(\mathbf{q})=\pi_B(\mathbf{q}) = 0.9375$.
\begin{figure}
\centering
\tikzstyle{H-node}=[rectangle,draw=black,fill=white!30,inner sep=1.3mm]
\tikzstyle{B-node}=[circle,draw=blue,fill=blue!20,inner sep=2.5mm]
\tikzstyle{G-node}=[circle,draw=green,fill=green!30,inner sep=1.3mm]
\tikzstyle{R-node}=[rectangle,draw=red,fill=red!20,inner sep=2.6mm]
\tikzstyle{W-node}=[rectangle,draw=white,fill=white!30,inner sep=0.2mm]
\tikzstyle{test-node}=[circle,draw=black,fill=black,inner sep=.2mm]

\tikzstyle{bl0} = [draw=black, thick, dashed]   
\tikzstyle{b9} = [draw=black, thick]   
\tikzstyle{bl1} = [->, draw=black]   
\tikzstyle{bl2} = [draw=black!70,thick]   
\tikzstyle{bl3} = [draw=black,thick, dotted]   

\tikzstyle{br0} = [draw=brown, dashed]   
\tikzstyle{br1} = [->, draw=brown]   
\tikzstyle{br2} = [->, draw=brown,thick]   

\tikzstyle{red0} = [draw=red, thick, dashed]   
\tikzstyle{red1} = [draw=red]   
\tikzstyle{red2} = [draw=red,thick]   

\tikzstyle{gr0} = [draw=green, thick, dashed]   
\tikzstyle{gr1} = [draw=green]   
\tikzstyle{gr2} = [draw=green,thick]   
\tikzstyle{gr4} = [draw=green,semithick,rounded corners]   

\begin{tikzpicture}[scale=0.5][domain=0:8]

\draw (1,5) node[B-node,label=center:$A$,label=above:$\pi_A$\text{$=$}$0.0938$] (b_1) {};
\draw (5,5) node[B-node,label=center:$B$,label=above:$\pi_B$\text{$=$}$0.0938$] (b_2) {};
\draw (3,0) node[R-node,label=center:$1$, label=below:$p_1$\text{$=$}$\frac{1}{2}$] (h_1) {};

\draw[b9] (b_1)  to node [label=left:$q_{A1}$] {} (h_1) ;
\draw[b9] (b_2)  to node [label=right:$q_{B1}$] {} (h_1) ;
\draw (3, 8) node[W-node,label=center:First Scenario] (s1) {};

\draw (9,5) node[B-node,label=center:$A$,label=above:$\pi_A$\text{$=$}$0.0938$] (bb_1) {};
\draw (13,5) node[B-node,label=center:$B$,label=above:$\pi_B$\text{$=$}$0.0938$] (bb_2) {};
\draw (9,0) node[R-node,label=center:$1$, label=below:$p_1$\text{$=$}$\frac{1}{2}$] (hh_1) {};
\draw (13,0) node[R-node,label=center:$2$, label=below:$p_2$\text{$=$}$\frac{1}{2}$] (hh_2) {};

\draw[b9] (bb_1) to node [label=left:$q_{A1}$] {} (hh_1) ;
\draw[b9] (bb_1) -- (hh_2) ;
\draw[b9] (bb_2) -- (hh_1) ;
\draw[b9] (bb_2) to node [label=right:$q_{B2}$] {} (hh_2) ;
\draw (11, 8) node[W-node,label=center:Second Scenario] (s2) {};

\draw (9.8, 3) node[W-node,label=center:$q_{A2}$] (lab2) {};
\draw (12.2, 3) node[W-node,label=center:$q_{B1}$] (lab3) {};

\draw (17,5) node[B-node,label=center:$A$,label=above:$\pi_A$\text{$=$}$0.124$] (bbb_1) {};
\draw (21,5) node[B-node,label=center:$B$,label=above:$\pi_B$\text{$=$}$0.064$] (bbb_2) {};
\draw (17,0) node[R-node,label=center:$1$, label=below:$p_1$\text{$=$}$0.64$] (hhh_1) {};
\draw (21,0) node[R-node,label=center:$2$, label=below:$p_2$\text{$=$}$0.48$] (hhh_2) {};

\draw[b9] (bbb_1)  to node [label=left:$q_{A1}$] {} (hhh_1) ;
\draw[b9] (bbb_1)  to node [label=left:$q_{A2}$] {} (hhh_2) ;
\draw[b9] (bbb_2)  to node [label=right:$q_{B2}$] {} (hhh_2) ;
\draw (19, 8) node[W-node,label=center:Third Scenario] (s3) {};

\end{tikzpicture}
\caption{This figure represents the three scenarios of our example. 
Vector $\mathbf{q}=(\frac{1}{4}, \frac{1}{4})$ represents the unique equilibrium in
the first scenario. Vector $\mathbf{q}=(\frac{1}{8}, \frac{1}{8}, \frac{1}{8}, \frac{1}{8})$ 
is the unique equilibrium of the second scenario. Finally, 
Vector $\mathbf{q}=(0.18, 0.1, 0.16)$ is the unique equilibrium in the third scenario.}
\label{fig1}
\end{figure}
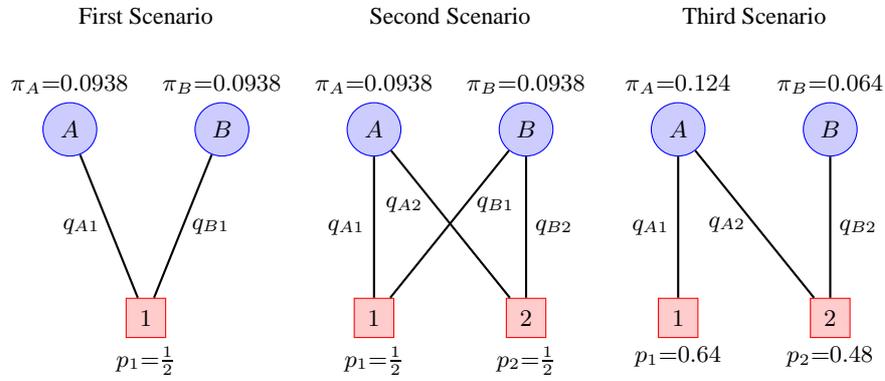

\item[Scenario 2]
We construct the second scenario by splitting the market in the previous scenario into 
two identical markets such that both firms have access to both markets (see Figure \ref{fig1}). 
Since the demand is divided between two identical markets,   
the price for market $j$ would be $p_j(\mathbf{q}) = 1 - 2q_{Aj} - 2q_{Bj}$, i.e., the clearance price of each market is the same as the clearance price of the market in Scenario 1, when the supply is half of the supply of the market in Scenario 1.
In this scenario, the profit of firm $i \in  \{A,B\}$ is 
$\pi_i(\mathbf{q}) = \sum_j q_{ij}(1-2q_{Aj}-2q_{Bj}) - \frac{1}{2}(q_{i1}+q_{i2})^2$. 
Any Nash equilibrium of this game satisfies the set of equations 
$\frac{\partial \pi_A}{\partial q_{A1}}=\frac{\partial \pi_A}{\partial q_{A2}}=\frac{\partial \pi_B}{\partial q_{B1}}=\frac{\partial \pi_B}{\partial q_{B2}}=0$.
%
~By finding the unique solution to this set of equations, one can verify that $\mathbf{q}=(\frac{1}{8}, \frac{1}{8}, \frac{1}{8}, \frac{1}{8})$ is the 
unique equilibrium of the game where  $p_1(\mathbf{q})=p_2(\mathbf{q})=\frac{1}{2}$, 
and $\pi_A(\mathbf{q})=\pi_B(\mathbf{q}) = 0.09375$. Since we artificially split the market 
into two identical markets, this equilibrium is, not surprisingly, the same as the equilibrium 
in the previous scenario.

\item[Scenario 3]
Consider the previous scenario, and suppose firm $2$ has no access to the first market 
(see Figure~\ref{fig1}). Let the demand functions and the cost functions be the same as the 
previous scenario. The profits of firms $1$ and $2$ can be written as follows:
\begin{eqnarray*}
\pi_A(\mathbf{q}) & = & q_{A1}(1-2q_{A1}) + q_{A2}(1-2q_{A2}-2q_{B2})- \frac{1}{2} (q_{A1}+q_{A2})^2, \\
\pi_B(\mathbf{q}) & = & q_{B2}(1-2q_{A2}-2q_{B2})- \frac{1}{2} q_{B2}^2. 
\end{eqnarray*}
The unique equilibrium outcome of the game is found by solving the set of equations 
$\frac{\partial \pi_A}{\partial q_{A1}}=\frac{\partial \pi_A}{\partial q_{A2}}=\frac{\partial \pi_B}{\partial q_{B2}}=0$. 
One can verify that vector $\mathbf{q}=(q_{A1}, q_{A2}, q_{B2})=(0.18, 0.1, 0.16)$ is the unique 
equilibrium outcome of the game where $p_1(\mathbf{q})=0.64$, $p_2(\mathbf{q})=0.48$, 
$\pi_A(\mathbf{q})=0.124$, and $\pi_B(\mathbf{q}) = 0.064$. 
The following are a few observations worth mentioning.
\begin{itemize}
\item Firm $A$ has more power in this scenario due to having a captive market\footnote{A captive market is one in which consumers 
have limited options and the seller has a monopoly power.}. 
\item The equilibrium price of market $1$ is higher than the equilibrium price 
in the previous scenarios.
\item The position of firm $B$ affects its profit. Since it has no access to market $1$, 
it is not as powerful as firm $A$.
\item The equilibrium price of market $2$ is smaller than the equilibrium price 
in the previous scenarios.
\end{itemize}
\end{description}
 
\subsection{Related Work}
There are several papers that investigate the Cournot competition 
in an oligopolistic setting (see, e.g., \cite{KS83,SV84,H90,H00,V01}). 
In spite of these works, little is known about the Cournot competition in a network. 
\citet{I09} studies the Cournot competition in a network setting, 
and considers a network of firms and markets where each firm chooses a quantity to supply 
in each accessible market.
He studies the competition when the inverse demand functions are linear and 
the cost functions are quadratic (functions of the total production). 
In this study, we consider the same model when the cost functions and 
the demand functions may have quite general forms. We show the game with 
linear inverse demand functions is a potential game and therefore has 
a unique equilibrium outcome. 
Furthermore, we present two polynomial-time algorithms for finding an equilibrium outcome for a wide range of cost functions and demand functions.
While we investigate the Cournot competition in networks, there is a recent paper which considers the Bertrand competition in network setting \cite{BLN13}, albeit in a much more restricted case of only two firms competing in each market.

The final price of each market in the Cournot competition is a market clearing price; i.e, the final price is set such that the market becomes clear. 
Finding a market clearance equilibrium is a well-established problem, and  there are several papers which propose polynomial-time algorithms for computing equilibriums of markets in which the price of each good is defined as the price in which the market clears. Examples of such markets include Arrow-Debreu market and its special case Fisher market (see related work on these markets \cite{clear7,clear3,clear4,clear2,clear1,clear5,clear6}). 
\citet{clear4} design an approximation scheme which computes the market clearing prices for the Arrow-Debreu market, and \citet{clear6} improve the running time of the algorithm. 
The first polynomial-time algorithm for finding an Arrow-Debreu market equilibrium is proposed by \citet{clear2} for a special case with linear utilities. The Fisher market, a special case of the Arrow-Debreu market, attracted a lot of attention as well. \citet{clear7} present the first polynomial-time algorithm by transferring the problem to a concave cost maximization problem. \citet{clear3} design the first combinatorial algorithm which runs in polynomial time and finds the market clearance equilibrium when the utility functions are linear. This result is later improved by \citet{clear5}.

For the sake of completeness, we refer to recent works in the computer science literature \cite{wine10,soda12}, which investigate the Cournot competition in an oligopolistic setting.
\citet{wine10} study a coalition formation game in a  Cournot oligopoly. In this setting, firms form coalitions, and the utility of each coalition, which is equally divided between its members, is determined by the equilibrium of a Cournot competition between coalitions. They prove the price of anarchy, which is the ratio between the social welfare of the worse stable partition and the social optimum, is $\Theta(n^{2/5})$ where $n$ is the number of firms.
\citet{soda12} consider a Cournot competition where agents may decide to be non-myopic. 
In particular, they define two principal strategies to maximize revenue and profit 
(revenue minus cost) respectively. Note that in the classic Cournot competition all agents 
want to maximize their profit. However, in their study each agent first chooses 
its principal strategy and then acts accordingly. The authors prove this game has a 
pure Nash equilibrium and the best response dynamics will converge to an equilibrium. 
They also show the equilibrium price in this game is lower than the equilibrium price 
in the standard Cournot competition.

\subsection{Results and techniques}
We consider the problem of Cournot competition on a network of markets and firms for different classes of cost and inverse demand functions. Adding these two dimensions to the classical Cournot competition which only involves a single market and basic cost and inverse demand functions yields an engaging but complicated problem which needs advanced techniques for analyzing.
For simplicity of notation we model the competition by a bipartite graph rather than a hypergraph: vertices on one side denote the firms, and vertices on the other side denote the markets. An edge between a firm and a market demonstrates willingness of the firm to compete in that specific market. The complexity of finding the equilibrium, in addition to the number of markets and firms, depends on the classes that inverse demand and production cost functions belong to.

We summarize our results in the following table.
\\[0mm]
\noindent
\begin{tabular}{|p{3.8cm}|p{3cm}|c|p{3cm}|}
\hline
\textbf{Cost functions} & \textbf{Inverse demand functions} & \textbf{Running time} & \textbf{Technique} \\
\hline
Convex&
Linear&
\parbox{1.9cm}{\qquad\\$O(E^3)$}&
Convex optimization, formulation as an ordinal potential game\\
\hline
Convex & 
Strongly monotone marginal revenue function\footnotemark& 
\parbox{1.9cm}{\qquad\\\qquad\\poly($E$)}&
Reduction to a nonlinear complementarity problem

\\ \hline
Convex, separable&
Concave&
$O(n\log^2Q_{\max})$&
Supermodular optimization, nested binary search
\\\hline
\end{tabular}
\footnotetext{Marginal revenue function is the vector function which maps production quantities for an edge to marginal revenue along that edge.}
In the above table, $E$ denotes the number of edges of the bipartite graph, $n$ denotes the number of firms, and $Q_{\max}$ denotes the maximum possible total quantity in the oligopoly network at any equilibrium. In our results we assume the inverse demand functions are nonincreasing functions of total production in the market. This is the basic assumption in the classical Cournot Competition model: As the price in the market increases, it is reasonable to believe that the buyers drop out of the market and demand for the product decreases. The classical Cournot Competition model as well as many previous works on Cournot Competition model assumes linearity of the inverse demand function~\cite{I09,wine10}. In fact there is little work on generalizing the inverse demand function in this model. The second and third row of the above table shows we have developed efficient algorithms for more general inverse demand functions satisfying concavity rather than linearity. This can be accounted as a big achievement. The assumption of monotonicity of the inverse demand function is a standard assumption in Economics~\cite{efficiency2003,amir1996cournot,milgrom1990rationalizability}. We assume cost functions to be convex which is the case in many works related to both Cournot Competition and Bertrand Network \cite{kukushkin1993cournot,weibull2006price}. In a previous work \cite{I09}, the author considered Cournot Competition on a network of firms and markets; however, assumed that inverse demand functions are linear and all the cost functions are quadratic function of the total production by the firm in all markets which is quite restrictive. Most of the results in other related works in Cournot Competition and Bertrand Network require linearity of the cost functions \cite {BLN13,wine10}. A brief summary of our results presented in three sections is given below.

\subsubsection{Linear Inverse Demand Functions}
In case inverse demand functions are linear and production costs
are convex, we present a fast and efficient algorithm to obtain the equilibrium.
This approach works by showing that Network Cournot Competition belongs to a class of games called \emph{ordinal potential games}. In such games, the collective strategy of the independent players
is to maximize a single potential function.
The potential function is carefully designed in such a way
that changes made by one player reflects in the same way in
the potential function as in their own utility function.
We design a potential function for the game, which depends on the network structure, and show how it captures this property. Moreover, in the case where the cost functions are convex, we prove concavity of this designed potential function (Theorem \ref{convex-potential}) concluding convex optimization methods can be employed to
find the optimum and hence, the equilibrium of the original Cournot competition. We also discuss uniqueness of equilibria in case the cost functions are strictly concave. Our result in this section is specifically interesting since we find the unique equilibrium of the game.
We prove the following theorems in Section~\ref{potential}.

\begin{theorem}
The Network Cournot Competition with linear inverse demand functions forms an ordinal potential game.
\end{theorem}

\begin{theorem}
Our designed potential function for the Network Cournot Competition with linear inverse demand functions is concave provided that the cost functions are convex. Furthermore, the potential function is strictly concave if the cost functions are strictly convex, and hence the equilibria for the game is unique. In addition, a polynomial-time algorithm finds the optimum of the potential function which describes the market clearance prices.
\end{theorem}

\subsubsection{The general case}
Since the above approach does not work for nonlinear inverse demand functions, we design another interesting but more involved algorithm to capture more general forms of inverse demand functions. We show that an equilibrium of the game can be computed in polynomial time
if the production cost functions are convex and the revenue function is monotone. Moreover, we show under strict monotonicity of the revenue function, the solution is unqiue, and therefore our results in this section is structural; i.e. we find the one and only equilibria. For convergence guarantee we also need Lipschitz condition on derivatives of inverse demand and cost functions. We start the section by modeling our problem as a complementarity problem. Then we prove how holding the aforementioned conditions for cost and revenue functions yields satisfying \emph{Scaled Lipschitz Condition} (SLC) and semidefiniteness for matrices of derivatives of the profit function. SLC is a standard condition widely used in convergence analysis for scalar and vector optimization \cite{zhao1999two}.
Finally , we present our algorithm, and show how meeting these new conditions by inverse demand and cost functions helps us to guarantee polynomial running time of our algorithm. We also give examples of classes of inverse demand functions satisfying the above conditions. These include many families of inverse demand functions including quadratic functions, cubic functions and entropy functions.
The following theorem is the main result of Section~\ref{general} which summarizes the performance of our algorithm.

\begin{theorem}\label{thm:result:1}
A solution to the Network Cournot Competition can be found in polynomial number of iterations under the following conditions:
\begin{enumerate}\setlength{\itemindent}{.2in}
\item The cost functions are (strongly) convex.
\item The marginal revenue function is (strongly\footnote{For at least one of the first two conditions, strong version of condition should be satisfied, i.e., either cost functions should be strongly convex or the marginal revenue function should be strongly monotone.}) monotone.
\item Ther first derivative of cost functions and inverse demand functions and the second derivative of inverse demand functions are Lipschitz continuous.
\end{enumerate}

Furthermore, the solution is unique assuming only the first condition. Therefore, our algorithm finds the unique equilibria of NCC.
\end{theorem}

\subsubsection{Cournot oligopoly}
Another reasonable model for considering cost functions of the firms is the case where the cost of production 
in a market depends only on the quantity produced by the firm in that specific market (and not on quantities produced by this firm in other markets). In other words, the firms have completely independent sections for producing different goods in various markets, and there is no correlation between cost of production in separate markets. Interestingly, in this case the competitions are separable; i.e. equilibrium for Network Cournot Competition can be found by finding the quantities at equilibrium for each market individually. This motivates us for considering Cournot game where the firms compete over a single market. We present a new algorithm for computing equilibrium quantities produced by firms in a Cournot oligopoly, i.e., when the firms compete over a single market. Cournot Oligopoly is a well-known model in Economics, and computation of its Cournot Equilibrium has been subject to a lot of attention. It has been considered in many works including \cite{thorlund1990iterative, kolstad1991computing, okuguchi1985existence, mathiesen1985computation, campos2010solving} to name a few. The earlier attempts for calculating equilibrium for a general class of inverse demand and cost functions are mainly based on solving a Linear Complementarity Problem or a Variational Inequality. These settings can be then turned into  convex optimization problems of size $O(n)$ where $n$ is the number of firms. This means the runtime of the earlier works cannot be better than $O(n^3)$ which is the best performance for convex optimization \cite{boydbook}. We give a novel combinatorial algorithm for this important problem when the quantities produced are integral. We limit our search to integral quantities for two reasons. First, in real-world all commodities and products are traded in integral units. Second, this algorithm can easily be adapted to compute approximate Cournot-Nash equilibrium for the continuous case and since the quantities at equilibrium may not be rational numbers, this is the best we can hope for. Our algorithm runs in time $O(n\log^2 (Q_{\max}))$ where $Q_{\max}$ is an upper bound on total quantity produced at equilibrium. Our approach relies on the fact that profit functions are supermodular when the inverse demand function is nonincreasing and the cost functions are convex. We leverage the supermodularity of inverse demand functions and Topkis' Monotonicity Theorem \cite{topkis1978minimizing}  to design a nested binary search algorithm. The following is the main result of Section \ref{sec-1}.

\begin{theorem}
A polynomial-time algorithm successfully computes the quantities produced by each firm at an equilibrium of 
the Cournot oligopoly if the inverse demand function is non-increasing, 
and the cost functions are convex. In addition, the algorithm runs 
in $O(n \log^2(Q_{\max}))$ where $Q_{\max}$ is the maximum possible 
total quantity in the oligopoly network at any equilibrium.
\end{theorem}

\section{Notations}

Suppose we have a set of $n$ firms denoted by $\mathcal{F}$ and 
a set of $m$ markets denoted by $\mathcal{M}$. 
A single good is produced in each of these markets. 
Each  firm might or might not be able to supply a particular market. 
A bipartite graph is used to demonstrate these relations. 
In this graph, the markets are denoted by the numbers $1,2,\dots,m$ on one side, 
and the firms are denoted by the numbers $1,2,\dots,n$ on the other side. 
For simplicity, throughout the paper we use the notation $i \in \mathcal{M}$ 
meaning the market $i$, and $j \in \mathcal{F}$ meaning firm $j$. 
For firm $j \in \mathcal{F}$ and market $i \in \mathcal{M}$ 
there exists an edge between the corresponding vertices in the bipartite graph 
if and only if firm $j$ is able to produce the good in market $i$. 
This edge will be denoted $(i,j)$. 
The set of edges of the graph is denoted by $\mathcal{E}$, 
and the number of edges in the graph is shown by $E$. 
For each market $i \in \mathcal{M}$, the set of vertices $N_\mathcal{M}(i)$ is 
the set of firms that this market is connected to in the graph. 
Similarly, $N_\mathcal{F}(j)$ denotes the set of neighbors of firms $j$ among markets. 
The edges in $\mathcal{E}$ are sorted and numbered $1,\dots,E$, 
first based on the number of their corresponding market and 
then based on the number of their corresponding firm. 
More formally, edge $(i,j) \in \mathcal{E}$ is ranked above 
edge $(l,k) \in \mathcal{E}$ if $i < l$ or $i = l$ and $j < k$. 
The quantity of the good that firm $j$ produces in market $i$ is denoted by $q_{ij}$. 
The vector $\mathbf{q}$ is an $E \times 1$ vector that contains all the quantities 
produced over the edges of the graph in the same order that the edges are numbered.

The demand for good $i$, denoted $D_{i}$, is the sum of the total quantity of this good 
produced by all firms, i.e., $D_{i} = \sum_{j\in N_\mathcal{M}(i)}{q_{ij}}$. 
The price of good $i$, denoted by the function $P_i(D_i)$, is only a decreasing function of total demand 
for this good and not the individual quantities produced by each firm in this market. For a firm $j$, the vector $\vec{s_j}$ denotes the strategy of firm $j$, 
which is the vector of all quantites produced by this firm in the markets $N_{\mathcal{F}}(j)$. 
Firm $j \in \mathcal{F}$ has a cost function related to its strategy denoted by $c_j(\vec{s_j})$. 
The profit that firm $j$ makes is equal to the total money that it obtains 
by selling its production minus its cost of production. 
More formally, the profit of firm $j$, denoted by $\pi_j$, is 
\begin{equation}
\pi_{j} = \sum_{i \in N_\mathcal{F}(j)} {P_i(D_i) q_{ij}}-c_{j}(\vec{s_j}).
\end{equation}

\section{Cournot competition and potential games}\label{potential}

In this section, we design an efficient algorithm for the case where the price functions are linear. More specifically, we design an innovative \emph{potential function} that captures the changes of all the utility functions simultaneously, and therefore, show how finding the quantities at the equilibrium would be equivalent to finding the set of quantities that maximizes this function. We use the notion of \emph{potential games} as introduced in \citet{monderer}. In that paper, the authors introduce \emph{ordinal potential games} as the set of games for which there exists a \emph{potential function} $P^*$ such that the pure strategy
equilibrium set of the game coincides with the pure strategy equilibrium set 
of a game where every party's utility function is $P^*$. 

In this section, we design a function for the Network Cournot Competition and show how this function is a potenial function for the problem if the price functions are linear. Interestingly, this holds for any cost function meaning Network Cournot Competition with arbitrary cost functions is an ordinal potential game as long as the price functions are linear.
Furthermore, we show when the cost functions are convex, our designed potential function is concave, and hence any convex optimization method can find the equilibrium of such a Network Cournot Competition. In case cost functions are strictly convex, the potential function is strictly concave. We restate a well known theorem in this section to conclude that the convex optimization in this case has a unique solution, and therefore the equilibria that we find in this case is the one and only equilibria of the game.

\begin{definition}
 A game is said to be an \emph{ordinal potential game} 
if the incentive of all players to change their strategy 
can be expressed using a single global function called 
the potential function. More formally, a game with $n$ players 
and utility function $u_i$ for player $i \in \{1,\dots,n\}$ 
is called ordinal potential with potential function $P^{*}$ 
if for all the strategy profiles $q \in \mathbb{R}^n$ and 
every strategy $x_i$ of player $i$ the following holds:
\begin{equation*}
u_{i}(x_i,q_{-i}) - u_{i}(q_i,q_{-i}) > 0 \;\;\mbox{iff} \;\;P^{*}(x_i,q_{-i}) - P^{*}(q_i,q_{-i}) > 0.
\end{equation*}

An equivalent definition of an ordinal potential game 
is a game for which a potential funciton $P^*$ exists 
such that the following holds for all strategy profiles 
$q \in \mathbb{R}^n$ and for each player $i$. 
\begin{align*}
\frac{\partial u_i}{\partial q_i} = \frac{\partial P^*}{\partial q_i}.
\end{align*}
In other words, for each strategy profile $q$, 
any change in the strategy of player $i$ has the same impact 
on its utility function as on the game's potential function.
\end{definition}

The pure strategy equilibrium set of any ordinal potential game coincides with the pure strategy
equilibrium set of a game with the potential function $P^*$ as all parties' utility function.

\begin{theorem}
\label{thm:pot:1}
The Network Cournot Competition with linear price functions is an ordinal potential game.
\end{theorem}
\begin{proof}
Let $P_i(D_i) = \alpha_i-\beta_iD_i$ be the linear price function for 
market $i \in \mathcal{M}$ where $\alpha_{i} \ge 0$ and $\beta_{i} \ge 0$ 
are constants determined by the properties of market $i$. Note that this 
function is decreasing with respect to $D_i$. Here we want to introduce 
a potential function $P^*$, and show that 
$\frac{\partial \pi_j}{\partial q_{ij}} = \frac{\partial P^*}{\partial q_{ij}}$ 
holds $\forall (i,j) \in \mathcal{E}$.
The utility function of firm $j$ is
\begin{align*}
\pi_{j} 
&= \sum_{i \in N_{\mathcal{F}}(j)} {\bigg(\alpha_i-\beta_i\sum_{k\in N_{\mathcal{F}}(j)}{q_{kj}}\bigg)q_{ij}}-c_{j}(\vec{s_j}),
\\\intertext{and taking partial derivative with respect to $q_{ij}$ yields}
\frac{\partial \pi_{j}}{\partial q_{ij}} 
&= \alpha_i-\beta_i\sum_{k\in N_{\mathcal{F}}(j)}{q_{kj}}-\beta_{i}q_{ij}-\frac{\partial c_{j}(\vec{s_j})}{\partial q_{ij}}.
\end{align*}

We define $P^{*}$ to be 
\begin{align*}
P^{*} 
&= \sum_{i\in \mathcal{M}}  {\bigg[\alpha_i\sum_{j\in N_{\mathcal{M}}(i)}{q_{ij}}   -  \beta_i\sum_{j\in N_{\mathcal{M}}(i)}{q_{ij}^2} - \beta_i\sum_{k\le j \atop k,j \in N_{\mathcal{M}}(i)}{q_{ij}q_{ik}} - \sum_{j\in N_{\mathcal{M}}(i)}{\frac{c_{j}(\vec{s_j})} {|N_{\mathcal{F}}(j)|}}\bigg]},
\\\intertext{whose partial derivative with respect to $q_{ij}$ is}
\frac{\partial P^{*}}{\partial q_{ij}} 
&= \alpha_i - 2\beta_i q_{ij} - \frac{\partial}{\partial q_{ij}} \left( \beta_i\sum_{l\le m \atop l,m \in N_{\mathcal{M}}(i)}{q_{il}q_{im}} \right) - \frac{\partial c_{j}(\vec{s_j})}{\partial q_{ij}}
\\ 
&= \alpha_i - 2\beta_i q_{ij} -  \beta_i(D_i-q_{ij} ) - \frac{\partial c_{j}(\vec{s_j})}{\partial q_{ij}} 
\\
&= \frac{\partial \pi_{j}}{\partial q_{ij}}.
\end{align*}

Since this holds for each $i \in \mathcal{M}$ and each $j \in \mathcal{F}$, 
the Network Cournot Competition is an ordinal potential game.
\end{proof}

We can efficiently compute the equilibrium of the game
if the potential function $P^*$ is easy to optimize.
Below we prove that this function is concave.

\begin{theorem}\label{convex-potential}
The potential function $P^{*}$ from the previous theorem is concave provided that the cost functions of the firms are convex. Moreover, if the cost functions are strictly convex then the potential function is strictly concave. 
\end{theorem}

\begin{proof}
%
%
The proof goes by decomposing $P^*$ into pieces that are concave.
We first define $f$ for one specific market $i$ as
\begin{align*}
f = \sum_{j\in N_{\mathcal{M}}(i)}{q_{ij}^2} +\sum_{k\le j \atop k,j \in N_{\mathcal{M}}(i)}{q_{ij}q_{ik}},
\end{align*}
and prove that it is convex.

Recall that $\mathbf{q}$ is an $E \times 1$ vector of all the quantities 
of good produced over the existing edges of the graph. 
We can write $f = \mathbf{q}^T M\mathbf{q}$ 
where $M$ is an $E\times E$ matrix with all elements 
on its diagonal equal to $1$ and all other elements equal to $\frac{1}{\sqrt{2}}$:
\[M = \begin{bmatrix}
1&\frac{1}{\sqrt{2}}&\cdots &\frac{1}{\sqrt{2}} \\
\frac{1}{\sqrt{2}}&1&\cdots &\frac{1}{\sqrt{2}}\\
\vdots & \vdots & \ddots & \vdots\\
\frac{1}{\sqrt{2}}&\frac{1}{\sqrt{2}}&\cdots &1
\end{bmatrix}.\]

To show that $f$ is convex, it suffices to prove that $M$ is positive semidefinite,
by finding a matrix $R$ such that $M = R^T R$. 
Consider the following $(E+1)\times E$ matrix:
\[R = \begin{bmatrix}
c&c&\cdots &c \\
a&0&\cdots &0 \\
0&a&\cdots &0 \\
\vdots & \vdots & \ddots & \vdots\\
0&0&\cdots &a
\end{bmatrix},\]
where $a, c$ are set below.
Let $R_i$ be the $i$-th column of $R$.
We have $R_i \cdot R_i = a^2+c^2$ and $R_i \cdot R_j = c^2$ for $i \neq j$.
%

Setting $c = 2^{-\frac{1}{4}}$ and $a = \sqrt{1-c}$ yields $M = R^TR$,
showing that $M$ is positive semidefinite, hence the convexity of $f$.
%

The following expression for a fixed market $i \in \mathcal{M}$,
sum of three concave functions, is also concave.
\begin{align*}
\alpha_i\sum_{j\in N_{\mathcal{M}}(i)}{q_{ij}}   -  \beta_i\bigg(\sum_{j\in N_{\mathcal{M}}(i)}{q_{ij}^2} +\sum_{k\le j \atop k,j \in N_{\mathcal{M}}(i)}{q_{ij}q_{ik}}\bigg) - \sum_{j\in N_{\mathcal{M}}(i)}{\frac{c_{j}(\vec{s_j})} {|N_{\mathcal{F}}(j)|}}.
\end{align*}

%
Summing over all markets proves concavity of $P^*$.
Note that if a function is the sum of a concave function and a strictly concave function, then it is strictly concave itself. Therefore, since $f$ is concave, we can conclude strictly concavity of $P^*$ under the assumption that the cost functions are strictly convex.
\end{proof}

The following well-known theorem discusses the uniqueness of the solution to a convex optimization problem.

\begin{theorem}\label{unique}
Let $F : \mathcal{K} \rightarrow \mathbb{R}^n$ be a strictly concave and continuous function for some finite convex space $\mathcal{K} \in \mathbb{R}^n$. Then the following convex optimization problem has a unique solution.

\begin{equation}
\max f(x) ~~s.t.~~ x\in \mathcal{K}.
\end{equation}
\end{theorem}

By Theorem \ref{convex-potential}, if the cost functions are strictly convex then the potential function is strictly concave and hence, by Theorem \ref{unique} the equilibrium of the game is unique.

Let $ConvexP(\mathcal{E},(\alpha_1,\dots,\alpha_m),(\beta_1,\dots,\beta_m),(c_1,\dots,c_n))$ 
be the following convex optimization program:
\begin{align}
\min ~~-\hspace{-1mm}\sum_{i\in \mathcal{M}}  \bigg[\alpha_i\hspace{-4mm}\sum_{j\in N_{\mathcal{M}}(i)}\hspace{-3mm}{q_{ij}}   -  \beta_i\hspace{-4mm}\sum_{j\in N_{\mathcal{M}}(i)}\hspace{-3mm}{q_{ij}^2} &- \beta_i\hspace{-4mm}\sum_{k\le j \atop k,j \in N_{\mathcal{M}}(i)}\hspace{-3mm}{q_{ij}q_{ik}} - \hspace{-4mm}\sum_{j\in N_{\mathcal{M}}(i)}{\frac{c_{j}(\vec{s_j})} {|N_{\mathcal{F}}(j)|}}\bigg]\\
\nonumber \text{subject to}\qquad\qquad
q_{ij} &\ge 0 ~~\forall (i,j) \in \mathcal{E}.
\end{align}

Note that in this optimization program we are trying to maximize $P^*$ 
for a bipartite graph with set of edges $\mathcal{E}$, linear price functions 
characterized by the pair $(\alpha_i,\beta_i)$ for each 
market $i \in \mathcal{M}$, and cost functions $c_j$ for each firm $j \in \mathcal{F}$. 
\begin{algorithm}[H]\label{potentialalg}
\caption{Compute quantities at equilibrium for the Network Cournot Competition.}
\begin{algorithmic}
  \Procedure{COURNOT-POTENTIAL}{$\mathcal{E}$, $c_j$, $(\alpha_i,\beta_i)$} \Comment Set of edges, cost functions and price functions
  \State Use a convex optimization algorithm to solve $$ConvexP(\mathcal{E},(\alpha_1,\dots,\alpha_m),(\beta_1,\dots,\beta_m),(c_1,\dots,c_n)).$$ and return the vector $\mathbf{q}$ of equilibrium quantities.
  \EndProcedure
\end{algorithmic}
\label{alg:potential}
\end{algorithm}

The above algorithm has a time complexity equal to the time complexity 
of a convex optimization algorithm with $E$ variables. The best such algorithm 
has a running time $O(E^3)$\cite{boydbook}.



\section{Finding equilibrium for cournot game with general cost and inverse demand functions}\label{general}

In this section, we formulate an algorithm for a much more general class of price and cost functions.
Our algorithm is based on reduction of Network Cournot Competition (NCC) to a polynomial time solvable class of
Non-linear Complementarity Problem (NLCP). First in Subsection \ref{avval}, we introduce our marginal profit function as the vector of partial derivatives of all firms with respect to the quantities that they produce. Then in Subsection \ref{dovvom}, we show how this marginal profit function can help us to reduce NCC to a general NLCP. We also discuss uniqueness of equilibrium in this situation which yields the fact that solving NLCP would give us the one and only equilibrium of this problem. Unfortunately, in 
its most general form, NLCP is computationally intractable. However, for a large class of functions,
these problems are polynomial time solvable. Most of the rest of this section is dedicated to proving the fact that NCC is polytime solvable on vast and important array of price and cost functions. In Subsection \ref{sevvom}, we rigorously define the conditions under which NLCP is polynomial time solvable. We present our algorithm in this subsection along with a theorem which shows it converges in polynomial number of steps. To show the conditions that we introduce for convergence of our algorithm in polynomial time are not restrictive, we give a discussion in Subsection \ref{charom} on the functions satisfying these conditions, and show they hold for a wide range of price functions.

\paragraph{Assumptions} 
\label{par:assumptions}
Throughout the rest of this section we assume that the price functions are decreasing and concave and the cost functions are strongly convex (\textit{The notion of strongly convex is to be defined later}). We also assume that for each firm there is a finite quantity at which extra production ceases to be profitable even if that is the only firm operating in the market. Thus, all production quantities and consequently quantities supplied to markets by firms are finite. In addition, we assume Lipschitz continuity and finiteness of the first and the second derivatives of price and cost functions. We note that these Lipschitz continuity assumptions are very common for convergence analysis in convex optimization~\cite{boydbook} and finiteness assumptions are implied by Lipschitz continuity. In addition, they are not very restrictive as we don't expect unbounded fluctuation in costs and prices with change in supply.
For sake of brevity, we use the terms inverse demand function and price function interchangeably.

\subsection{Marginal profit function}\label{avval}
For the rest of this section, we assume that $P_i$ and $c_i$ are twice differentiable functions of quantities.
We define $f_{ij}$ for a firm $j$ and a market $i$ such that $(i,j)\in \mathcal{E}$  as follows.
\begin{align}\label{fij}
f_{ij} &= -\frac{\partial \pi_{j}}{\partial q_{ij}} = -P_i(D_i)-\frac{\partial P_i(D_i)}{\partial q_{ij}}q_{ij} + \frac{\partial c_j}{\partial q_{ij}}.
\\ \intertext{%
Recall that the price function of a market is only 
a function of the total production in that market and 
not the individual quantities produced by individual firms. 
Thus $\frac{\partial P_i(D_i)}{\partial q_{ij}} = \frac{\partial P_i(D_i)}{\partial q_{ik}}~~\forall j,k\in N_{\mathcal{M}}(i)$.	
Therefore, we replace these terms by $P_i'(D_i)$.
}
&f_{ij} = -P_i(D_i)-P_i'(D_i)q_{ij} + \frac{\partial c_j}{\partial q_{ij}}.
\end{align}

Let   vector $F$ be the vector of all $f_{ij}$'s corresponding to 
the edges of the graph in the same format that we defined the vector ${q}$. 
That is $f_{ij}$ corresponding to $(i,j)\in \mathcal{E}$ appears above $f_{lk}$ 
corresponding to edge $(l,k)\in \mathcal{E}$ iff $i < l$ or $i = l$ and $j < k$. 
Note that $F$ is a function of ${q}$.

Moreover, we separate the part representing marginal revenue from the part representing 
marginal cost in function $F$. More formally, we split $F$ into two functions $R$ and $S$ 
such that $F = R+S$, and the element corresponding to the edge $(i,j) \in \mathcal{E}$ in the \emph{marginal revenue fuction} $R({q})$ is:
\begin{align*}
r_{ij} &=  -\frac{\partial \pi_{j}}{\partial q_{ij}} = -P_i(D_i)-P'(D_i)q_{ij},
\\\intertext{%
whereas for the \emph{marginal cost function} $S({q})$, it is:
}
s_{ij} &= \frac{\partial c_j}{\partial q_{ij}}.
\end{align*}

\subsection{Non-linear complementarity problem}\label{dovvom}
In this subsection we formally define the non-linear complementarity problem (NLCP), and prove our problem is a NLCP.
\begin{definition}
Let $F : \mathbb{R}^n \rightarrow \mathbb{R}^n$ be a continuously differentiable function 
on $\mathbb{R}_{+}^n$. The complementarity problem seeks a vector $x \in \mathbb{R}^n$ 
that satisfies the following constraints:
\begin{equation}\label{CP}
\begin{split}
x , F(x) \ge 0, \\
x^T F(x) = 0.
\end{split}
\end{equation}
\end{definition}

\begin{theorem}
The problem of finding the vector ${q}$ at equilibrium in the Cournot game 
is a complementarity problem.
\end{theorem}

\begin{proof}
Let ${q}^*$ be the vector of the quantities at equilibrium. 
All quantities must be nonnegative at all times; i.e., ${q}^* \ge 0$. 
It suffices to show $F({q}^*) \ge 0$ and ${{q}^*}^T F({q}^*) = 0$. 
At equilibrium, no party benefits from changing its strategy, 
in particular, its production quantities. 
For each edge $(i,j) \in \mathcal{E}$, if the corresponding quantity $q_{ij}^*$ is positive, 
then $q_{ij}^*$ is a local maxima for $\pi_j$; i.e., 
$f_{ij}({q}^*) = -\left.\frac{\partial \pi_j}{\partial q_{ij}} \right|_{{q}^*} = 0$. 
On the other hand, if $q_{ij}^* = 0$, then 
$\left.\frac{\partial \pi_j}{\partial q_{ij}} \right|_{{q}^*} $ cannot be positive, 
since, if it is, firm $j$ would benefit by increasing the quantity $q_{ij}$ 
to a small amount $\epsilon$. 
Therefore, $\left.\frac{\partial \pi}{\partial q_{ij}} \right|_{{q}^*}$ 
is always nonpositive or equivalently $f_{ij}({q}^*) \ge 0$, 
i.e.,  $F({q}^*) \ge 0$. 
Also, as we mentioned above, a nonzero $q_{ij}^*$ is a local maximum for $\pi_j$; i.e., 
$f_{ij}({q}^*) = -\left.\frac{\partial \pi}{\partial q_{ij}} \right|_{{q}^*} = 0$. 
Hence, either $q_{ij}^* = 0$ or $f_{ij}({q}^*) = 0$; 
thus, $q_{ij}^* f_{ij}({q}^*) = 0$. 
This yields $\sum_{(i,j) \in \mathcal{E}}{q_{ij}^* f_{ij}({q}^*)} = {{q}^*}^T F({q}^*) = 0$.
\end{proof}

\begin{definition}
$F : \mathcal{K} \rightarrow \mathbb{R}^n$ is said to be strictly monotone at $x^*$ if

\begin{equation}\label{strictlymonotone}
\langle (F(x)-F(x^*))^T, x-x^* \rangle \ge 0, \forall x \in \mathcal{K}.
\end{equation}

$F$ is said to be strictly monotone if it is monotone at any $x^* \in \mathcal{K}$. Equivalently, $F$ is strictly monotone if the jacobian matrix is positive definite.
\end{definition}

The following theorem is a well known theorem for Complementarity Problems.

\begin{theorem}
Let $F : \mathcal{K} \rightarrow \mathbb{R}^n$ be a continuous and strictly monotone function with a point $x \in \mathcal{K}$ such that $F(x) \ge 0$ (i.e. there exists a potential solution to the CP). Then the Complementarity Problem introduced in (\ref{CP}) characterized by function $F$ has a unique solution.
\end{theorem}

Hence, the Complementarity Problem characterized by function $F$ introduced element by element in (\ref{fij}) has a unique solution under the assumption that the revenue function is strongly monotone (special case of strictly monotone). Note that the marginal profit function or $F$ in our case is non-negative in at least one point. Otherwise, no firm has any incentive to produce in any market and the equilibrium is when all production quantities are equal to zero. In the next subsection, we aim to find this unique equilibrium of the NCC problem.

\subsection{Designing a polynomial-time algorithm}\label{sevvom}
In this subsection, we introduce Algorithm \ref{generalalg} for finding equilibrium of NCC, and show it converges in polynomial time by Theorem \ref{zhaotheorem}. This theorem requires the marginal profit function to satisfy Scaled Lipschitz Condition(SLC) and monotonicity. We first introduce SLC, and show how the marginal profit function satisfies SLC and montonicty by Lemmas \ref{separation} to \ref{SLC}. We argue the conditions that the cost and price functions should have in order for the marginal profit function to satisfy SLC and monotonicity in Lemma \ref{SLC}. Finally, in Theorem \ref{zhaotheorem}, we show convergence of our algorithm in polynomial time. 

Before introducing the next theorem, 
we explain what the Jacobians $\nabla R$, $\nabla S$, and $\nabla F$ are for the Cournot game. 
First note that these are $E \times E$ matrices. 
Let $(i,j) \in \mathcal{E}$ and $(l,k) \in \mathcal{E}$ be two edges of the graph. 
Let $e_1$ denote the index of edge $(i,j)$, and $e_2$ denote the index of edge $(l,k)$ 
in the graph as we discussed in the first section. 
Then the element in row $e_1$ and column $e_2$ of matrix $\nabla R$, 
denoted $\nabla R_{e_1e_2}$, is equal to $\frac{\partial r_{ij}}{\partial q_{lk}}$. 
We name the corresponding elements in $\nabla F$ and $\nabla S$ similarly.
We have $\nabla F = \nabla R + \nabla S$ as $F = R+S$.

\begin{definition}[Scaled Lipschitz Condition (SLC)] 
A function $G :D \mapsto \mathbb{R}^n$, $D \subseteq \mathbb{R}^n$ is said to satisfy \emph{Scaled Lipschitz Condition (SLC)} if there exists a scalar $\lambda >0$ such that $\forall~h \in \mathbb{R}^n, \forall~x \in D$, such that $\Vert{X^{-1}h}\Vert \le 1$, we have:
\begin{equation}
\Vert X[G(x+h)-G(x)-\nabla G(x)h]\Vert_\infty \le \lambda \vert h^T \nabla G(x) h \vert,
\end{equation}
where $X$ is a diagonal matrix with diagonal entries equal to elements of the vector $x$ 
in the same order, i.e., $X_{ii} = x_i$  for all $i \in \mathcal{M}$.
\end{definition}

Satisfying SLC and monotonicity are essential for marginal profit function in Theorem \ref{zhaotheorem}. In Lemma \ref{SLC} we discuss the assumptions for cost and revenue function under which these conditions hold for our marginal profit function. We use Lemmas \ref{separation} to \ref{SLC} to show $F$ satisfies SLC. More specifically, we demonstrate in Lemma \ref{separation}, if we can derive an upperbound for LHS of SLC for $R$ and $S$, then we can derive an upperbound for LHS of SLC for $F = R+S$ too. Then in Lemma \ref{SLCLHS_S} and Lemma \ref{SLCLHS_R} we show LHS of $S$ and $R$ in SLC definition can be upperbounded. Afterwards, we show monotonicity of $S$ in Lemma \ref{lem:smono}. In Lemma \ref{SLC} we aim to prove $F$ satifies SLC under some assumptions for cost and revenue functions. We use the fact that LHS of SLC for $F$ can be upperbounded using Lemma \ref{SLCLHS_R} and Lemma \ref{SLCLHS_S} combined with Lemma \ref{separation}. Then we use the fact that RHS of SLC can be upperbounded using strong monotonicity of $R$ and Lemma \ref{lem:smono}. Using these two facts, we conclude $F$ satisfies SLC in Lemma \ref{SLC}.

\begin{lemma}\label{separation}
Let $F,R,S$ be three $\mathbb{R}^n \rightarrow \mathbb{R}^n$ functions 
such that $F(q) = R(q) + S(q),\quad \forall q \in \mathbb{R}^n$.
Let $R$ and $S$ satisfy the following inequalities for some $C > 0$ and $\forall~h$ such that $\Vert X^{-1}h \Vert \leq 1$:
\begin{align*}
  \Vert X[R(q+h)-R(q)-\nabla R(q)h]\Vert_\infty &\le C \Vert h \Vert^2, \\
  \Vert X[S(q+h)-S(q)-\nabla S(q)h]\Vert_\infty &\le C \Vert h \Vert^2,
\end{align*}
where $X$ is the diagonal matrix with $X_{ii} = q_i$. 
Then we have:
\begin{align*}
  \Vert X[F(q+h)-F(q)-\nabla F(q)h]\Vert_\infty &\le 2C \Vert h \Vert^2.
\end{align*}
\end{lemma}

The following lemmas give upper bounds for LHS of the SLC for $S$ and $R$ respectively.
\begin{lemma}\label{SLCLHS_S}
Assume $X$ is the diagonal matrix with $X_{ii} = q_i$. 
  $\forall~h$ such that $\Vert X^{-1}h \Vert \leq 1$, there exists a constant $C > 0$ satisfying:
    $\Vert X[S(q+h)-S(q)-\nabla S(q)h]\Vert_\infty \le C \Vert h \Vert^2$.
\end{lemma}
\begin{lemma}\label{SLCLHS_R}
Assume $X$ is the diagonal matrix with $X_{ii} = q_i$. 
  $\forall~h$ such that $\Vert X^{-1}h \Vert \leq 1$, $\exists C > 0$ such that
    $\Vert X[R(q+h)-R(q)-\nabla R(q)h]\Vert_\infty \le C \Vert h \Vert^2$.
\end{lemma}

If $R$ is assumed to be strongly monotone, we immediately have a lower bound on RHS of the SLC for $R$.
The following lemma gives a lower bound on RHS of the SLC for $S$.
\begin{lemma}
  \label{lem:smono}
  If cost functions are (strongly) convex $S$ is (strongly) monotone\footnote{A matrix $M \in \mathbb{R}^{n \times n}$ is \emph{strongly positive definite} iff $\forall~x \in \mathbb{R}^n$ and some $\alpha > 0$ $x^TMx \geq \alpha\Vert x \Vert^2$.}\footnote{  A differentiable function $f:D \mapsto \mathbb{R}^n$ is \emph{monotone} iff its Jacobian $\nabla f$ is positive semidefinite over its domain $D$.}\footnote{  A differentiable function $f:D \mapsto \mathbb{R}^n$ is \emph{strongly monotone} iff its Jacobian $\nabla f$ is strongly positive definite over its domain, $D$.}\footnote{A twice differentiable function $f:D \mapsto \mathbb{R}$ is \emph{strongly convex} iff its Hessian $\nabla^2 f$ is strongly positive definite over its domain, $D$.}.
\end{lemma}

The following lemma combines the results of Lemma \ref{SLCLHS_S} and Lemma \ref{SLCLHS_R} using Lemma \ref{separation} to derive an upper bound for LHS of the SLC for $F$. We bound RHS of the SLC from below by using strong montonicity of $R$ and Lemma \ref{lem:smono}.
\begin{lemma} \label{SLC}
  $F$ satisfies SLC and is monotone if:
  \begin{enumerate}
  \item  Cost functions are convex. \item Marginal revenue function is monotone. \item Cost functions are strongly convex or marginal revenue function is strongly monotone.
  \end{enumerate}
\end{lemma}


We wrap up with the following theorem, which summarizes the main result of this section.
Lemma \ref{SLC} guarantees that our problem satisfies the two conditions 
mentioned in \citeauthor{zhao1999two}~\citeyear{zhao1999two}. Therefore, we can prove the following theorem.
\begin{theorem}\label{zhaotheorem}
Algorithm \ref{generalalg} converges to an equilibrium of Network Cournot Competition in time $O\big(E^2\log(\mu_0/\epsilon)\big)$ under the following assumptions:
\begin{enumerate}\setlength{\itemindent}{.2in}
\item The cost functions are strongly convex.
\item The marginal revenue function is strongly monotone.
\item The first derivative of cost functions and price functions and the second derivative of price functions are Lipschitz continuous.
\end{enumerate}
This algorithm outputs an approximate solution $(F(q^*),q^*)$ satisfying $(q^*)^T F(q^*)/n \le \epsilon$ where $\mu_0 = (q_0)^T F(q_0)/n$, and $(F(q_0),q_0)$ is the initial feasible point \footnote{Initial feasible solution can be trivially found. E.g., it can be the same production quantity along each edge, large enough to ensure losses for all firms. Such quantity can easily be found by binary search between [0, Q].}. 
\end{theorem}

\begin{algorithm}
\caption{Compute quantities at equilibrium for the Cournot game.}
\begin{algorithmic}[1]
\Procedure{NETWORK-COURNOT}{$P_i, c_j, \epsilon$}\Comment{The price function $P_i$ for each market $i \in \mathcal{M}$, the cost function $c_j$ for each firm $j\in \mathcal{F}$, and $\epsilon$ as the desired tolerance}
  \State Calculate vector $F$ of length $E$ as defined in \eqref{fij}.
  \State Find the initial feasible\footnotemark solution $(F(x_0),x_0)$ for the complementarity problem. This solution should satisfy $x_0 \ge 0$ and $F(x_0) \ge 0$.
  \State Run Algorithm $3.1$ from \cite{zhao1999two} to find the solution $(F(x^*),x^*)$ to the CP characterized by $F$.
  \State \Return $x^*$ \Comment{The vector $q$ of quantities produced by firms at equilibrium}
\EndProcedure
\end{algorithmic}
\label{generalalg}
\end{algorithm}

\subsection{Price Functions for Monotone Marginal Revenue Function}\label{charom}
This section will be incomplete without a discussion of price functions that satisfy the convergence conditions for Algorithm \ref{generalalg}. We will prove that a wide
variety of price functions preserve monotonicity of the marginal revenue function. To this end, we prove the following lemma.


\begin{lemma}\label{PSD_R}
$\nabla R({q})$ is a positive semidefinite matrix $\forall~{q} \geq 0$, i.e., $R$ is monotone, provided that for all markets $|P'_i(D_i)| \geq \frac{|P''_i(D_i)|D_i}{2}$.
\end{lemma}

While the above condition may seem somewhat restrictive, they allow the problem to be solved on a wide range of price functions. Intuitively, the condition implies that linear and quadratic terms dominate higher order terms. We present the following corollaries as examples of classes of functions that satisfy the above condition.

\begin{corollary}
  All decreasing concave quadratic price functions satisfy Lemma \ref{PSD_R}.
\end{corollary}
\begin{corollary}
  All decreasing concave cubic price functions satisfy Lemma \ref{PSD_R}.
\end{corollary}
\begin{corollary}
  Let $a_i \in \mathbb{R}_{\geq 0}^n$ for $i \in \{1\ldots k\}$ be arbitrary positive vectors. Let $f:\mathbb{R}^n_{\geq 0} \mapsto R$ be the following function: $f(x) = \sum_{i \in \{1\ldots k\}}(a_i^Tx)\log (a_i^Tx)$. Then $f$ (and $-f$) satisfies Lemma \ref{PSD_R}.
\end{corollary}

\section{Algorithm for Cournot Oligopoly} 
\label{sec-1}
In this section we present a new algorithm for computing equilibrium quantities produced by firms in a Cournot oligopoly, i.e., when the firms compete over a single market. Cournot Oligopoly is a standard model in Economics and computation of Cournot Equilibrium is an important problem in its own right. A considerable body of literature has been dedicated to this problem \cite{thorlund1990iterative, kolstad1991computing, okuguchi1985existence, mathiesen1985computation, campos2010solving}. All of the earlier works that compute Cournot equilibrium for a general class of price and cost functions rely on solving a Linear Complementarity Problem or a Variational Inequality which in turn are set up as convex optimization problems of size $O(n)$ where $n$ is the number of firms in oligopoly. Thus, the runtime guarantee of the earlier works is $O(n^3)$ at best. We give a novel combinatorial algorithm for this important problem when the quantities produced are integral. Our algorithm runs in time $n\log^2 (Q_{max})$ where $Q_{max}$ is an upper bound on total quantity produced at equilibrium. We note that, for two reasons, the restriction to integral quantities is practically no restriction at all. Firstly, in real-world all commodities and products are traded in integral units. Secondly, this algorithm can easily be adapted to compute approximate Cournot-Nash equilibrium for the continuous case and since the quantities at equilibrium may not be rational numbers, this is the best we can hope for.

As we have only a single market rather than a set of markets, we make a few changes to the notation. 
Let $[n] = \{1,\ldots,n\}$ be the set of firms competing over the single market. 
Let $\mathbf{q} = (q_1,q_2,\dots,q_n)$ be the set of all quantities produced by the firms. 
Note that in this case, each firm is associated with only one quantity. 
Let $Q = \sum_{i \in [n]}{q_i}$ be the sum of the total quantity of good produced in the market. 
In this case, there is only a single inverse demand function 
$P:\mathbb{Z} \mapsto \mathbb{R}_{\geq 0}$, which maps total supply, $Q$, to market price.
We assume that price decreases as the total quantity produced by the firms increases, i.e., $P$ is a decreasing function of $Q$. 
For each firm $i \in [n]$, the function $c_i:\mathbb{Z} \mapsto \mathbb{R}_{\geq 0}$ 
denotes the cost to this firm when it produces quantity $q_i$ of the good in the market. 
The profit of firm $i \in [n]$ as a function of $q_i$ and $Q$, denoted $\pi_i(q_i,Q)$, is $P(Q)q_i - c_i(q_i)$. 
Also let $f_i(q_i,Q) = \pi_i(q_i+1,Q+1)-\pi_i(q_i,Q)$ be the marginal profit for 
firm $i \in [n]$ of producing one extra unit of product. 
Although the quantities are nonnegative integers, 
for simplicity we assume the functions $c_i$, $P$, $\pi_i$ and $f_i$ 
are zero whenever any of their inputs are negative. Also, we refer to the forward difference $P(Q+1)-P(Q)$ by $P'(Q)$.

\subsection{Polynomial time algorithm}
We leverage the supermodularity of price functions and Topkis' Monotonicity Theorem \cite{topkis1978minimizing} (Theorem \ref{thm:topkis}) to design a nested binary search algorithm which finds the equilibrium 
quantity vector $\mathbf{q}$ when the price function is a decreasing function of $Q$ and the cost functions of the firms are convex. Intuitively the algorithm works as follows. 
At each point we guess $Q'$ to be the total quantity of good produced by all the firms. 
Then we check how good this guess is by computing for each firm the set of quantites 
that it can produce at equilibrium if we assume the total quantity is the fixed integer $Q'$. 
We prove that the set of possible quantities for each firm at equilibrium,
assuming a fixed total production, is a consecutive set of integers. 
Let $I_i=\{q_i^l, q_i^l+1,\ldots, q_i^u-1,q_i^u\}$ be the range 
of all possible quantities for firm $i \in [n]$ assuming $Q'$ is the 
total quantity produced in the market. 
We can conclude $Q'$ was too low a guess if $\sum_{i \in [n]}{q_i^l} > Q'$. 
This implies our search should continue among total quantities above $Q'$. 
Similarly, if $\sum_{i \in [n]}{q_i^u} < Q'$, we can conclude our guess was too high, 
and the search should  continues among total quantities below $Q'$. 
If neither case happens, then for each firm $i \in [n]$, 
there exists a $q'_i \in I_i$ such that $Q' = \sum_{i \in [n]}{q'_i}$ 
and firm $i$ has no incentive to change this quantity 
if the total quantity is $Q'$. 
This means that the set $\mathbf{q'} = \{q_1',\dots,q_n'\}$ 
forms an equilibrium of the game and the search is over.

\begin{algorithm}
  \begin{algorithmic}[1]
    \Procedure{COURNOT-OLIGOPOLY}{$P, c_i$} \Comment{The market price function $P$, the cost functions $c_i$ for each firm $i\in [n]$}
    \State Let $Q_{\min} := 1$
    \State Let $Q_i^*$ be the optimal quantity that is produced by a firm when it is the only firm in the market
    \State Let $Q_{\max} := \sum_{i \in [n]}Q_i^*$
    \While{$Q_{\min} \leq Q_{\max}$}
      \State $Q' := \lfloor{\frac{Q_{\min} + Q_{\max}}{2}}\rfloor$
      \ForAll{$i \in [n]$}
        \State Binary search to find the minimum nonnegative integer $q_i^l$ satisfying 
        \State $f_i(q_i^l,Q') = \pi_i(q_i^l+1, Q'+1) - \pi_i(q_i^l, Q') \leq 0$ \label{eq:alg1}
        \State Binary search to find the maximum integer $q_i^u \leq Q'+1$ satisfying  
        \State $f_i(q_i^u-1,Q'-1) = \pi_i(q_i^u, Q') - \pi_i(q_i^u-1, Q'-1) \geq 0$ \label{eq:alg2}
        \State Let $I_i = \{q_i^l,\dots,q_i^u\}$ be the set of all integers between $q_i^l$ and $q_i^u$.
      \EndFor
      \If{$\Sigma_{i \in [n]}q_i^l > Q' $}
        \State $Q_{\min} := Q'+1$
      \ElsIf{$\Sigma_{i \in [n]}q_i^u < Q'$}
        \State $Q_{\max} := Q'-1$
      \Else
        \State Find a vector of quantities $\mathbf{q}=(q_1,\ldots,q_n)$ such that $q_i \in I_i$ and $\sum_{i \in [n]} q_i = Q'$
        \State \textbf{return} $\mathbf{q}$
      \EndIf
    \EndWhile
    \EndProcedure
  \end{algorithmic}
    \caption{Compute quantities produced by firms in a Cournot oligopoly.}
  \label{alg:bin}
\end{algorithm}

The pseudocode for the algorithm is provided in Algorithm~\ref{alg:bin}, whose correctness we prove next. 
The rest of this section is dedicated to proving Theorem~\ref{thm:bin:main}. Here, we present a brief outline of the proof. To help with the proof we define the functions $F_i$ and $G_i$ as follows. Let $F_i(q_i, Q) = P(Q+1)q_i + \frac{P'(Q)}{2}(q_i-\frac{1}{2})^2 - c(q_i)$. We note that the first difference of $F(q_i, Q)$ is the marginal profit for firm $i$ for producing one more quantity given that the total production quantity is $Q$ and firm $i$ is producing $q_i$. Let $G_i(q_i, Q) = F_i(q_i, Q-1)$. The first difference of $G_i(q_i, Q)$ is the marginal loss for firm $i$ for producing one less quantity given that the total production quantity is $Q$ and firm $i$ is producing $q_i$. Maximizers of these functions are closely related to equilibrium quantities a firm can produce given that the total quantity in market is $Q$. We make this connection precise and prove the validity of binary search in Lines 8-12 of Algorithm~\ref{alg:bin} in Lemma~\ref{lem:bin:1}. In Lemma~\ref{lem:bin:2}, we prove that $F_i$ and $G_i$ are supermodular functions of $q_i$ and $-Q$. In lemmas \ref{lem:bin:3} and \ref{lem:bin:3a}, we use Topkis' Monotonicity Theorem to prove the monotonicity of maximizers of $F_i$ and $G_i$. This, along with lemmas \ref{lem:bin:4} and \ref{lem:bin:5} leads to the conclusion that the outer loop for finding total quatity at equilibrium is valid as well and hence the algorithm is correct.
\subsection{Proof of correctness}
Throughout this section \emph{we assume that the price function is decreasing and concave and the cost functions are convex}.
\begin{lemma}
  \label{lem:bin:1}
  Let $q_i^*(Q) = \{q_i^l\ldots q_i^u\}$ , where $q_i^l=\min argmax_{q_i \in \{0\ldots Q_{max}\}}F_i(q_i,Q)$ and $q_i^u=\max argmax_{q_i \in \{0\ldots Q_{max}\}}G_i(q_i,Q)$. Then $q_i^*(Q)$ is the set of consecutive integers $I_i$ given by binary search in lines 8-12 of Algorithm \ref{alg:bin}. This is the set of quantities firm $i$ can produce at equilibrium given that the total quantity produced is $Q$.
\end{lemma}

\begin{lemma}
  \label{lem:bin:2}
  Let $F_i^-(q_i, -Q) = F_i(q_i, Q)$ and $G_i^- = G_i(q_i, -Q)$. Then $F_i^-$  and $G_i^-$ are supermodular functions.
\end{lemma}

\begin{lemma}
  \label{lem:bin:3}
  Let $I=\{q_i^l,\ldots , q_i^u\} = q_i^F(Q) = argmax_{q_i \in \{1\ldots Q_{max}\}}F_i(q_i,Q)$ and $I'=\{q^{'l}_i,\ldots, q^{'u}_i\} = q_i^F(Q') = argmax_{q_i \in \{1\ldots Q_{max}\}}F_i(q_i,Q')$. Let $Q > Q'$. Then $q^{'l}_i \geq q_i^l$ and $q^{'u}_i \geq q_i^u$.
\end{lemma}

\begin{lemma}
  \label{lem:bin:3a}
  Let $I=\{q_i^l,\ldots , q_i^u\} = q_i^G(Q) = argmax_{q_i \in \{1\ldots Q_{max}\}}G_i(q_i,Q)$ and $I'=\{q^{'l}_i,\ldots, q^{'u}_i\} = q_i^G(Q') = argmax_{q_i \in \{1\ldots Q_{max}\}}F_i(q_i,Q')$. Let $Q > Q'$. Then $q^{'l}_i \geq q_i^l$ and $q^{'u}_i \geq q_i^u$.
\end{lemma}

\begin{lemma}
  \label{lem:bin:4}
 Let $Q$ be total production quantity guessed by Algorithm~\ref{alg:bin} at a step of outer binary search. Let $I=(I_1,\ldots,I_n)$, where $I_i = \{q_i^{l},\ldots,q_i^{u}\}$, be the set of best reponse ranges of all firms if the total quantity is a fixed integer $Q$. Then, if $\sum_{i=1}^n q_i^{u} < Q$, there does not exist any equilibrium for which the total produced quantity is greater than or equal to $Q$.
\end{lemma}

\begin{lemma}
  \label{lem:bin:5}
 Let $Q$ be total production quantity guessed by Algorithm~\ref{alg:bin} at a step of outer binary search. Let $I=(I_1,\ldots,I_n)$, where $I_i = \{q_i^{l},\ldots,q_i^{u}\}$, be the set of best reponse ranges of all firms if the total quantity is a fixed integer $Q$. Then, if $\sum_{i=1}^n q_i^{l} > Q$, there does not exist any equilibrium for which the total produced quantity is less than or equal to $Q$.
\end{lemma}

Finally, the results of this section culiminate in the following theorem.

\begin{theorem}
\label{thm:bin:main}
Algorithm~\ref{alg:bin} successfully computes the vector $\mathbf{q} = (q_1,q_2,\dots,q_n)$ of quantities at one equilibrium of the Cournot oligopoly if the price function is decreasing and concave and the cost function is convex. In addition, the algorithm runs in time $O(n \log^2(Q_{\max}))$ where $Q_{\max}$ is the maximum possible total quantity in the oligopoly network at any equilibrium.
\end{theorem}

\begin{proof}
Lemma \ref{lem:bin:1} guarantees that the inner binary search successfully finds 
the best response range for all firms. 
Lemmas~\ref{lem:bin:4} and \ref{lem:bin:5} ensure that 
the algorithm always continues its search for the total quantity 
at equilibrium in the segment where all the equilibria are. 
Thus, when the search is over, it must be at an equilibrium of the game 
if one exists. If an equilibrium does not exist, then the algorithm will stop when it has eliminated all quantities in $\{1 \ldots Q_{max}\}$ as possible total equilibrium production quantities.
Let $Q_{\max}$ be the maximum total quantity possible 
at any equilibrium of the oligopoly network. 
Our algorithm performs a binary search over all possible quantities in $[1,Q_{\max}]$, 
and at each step finds a range of quantities for each firm $i \in [n]$ 
using another binary search. 
This means the algorithm runs in time $O(n\log^2(Q_{\max}))$. 
We can find an upper bound for $Q_{\max}$, 
noting that $Q_{\max}$ is at most the sum of the production quantites 
of the firms when they are the only producer in the market; 
i.e, $Q_{\max} \le \sum_{i\in [n]}{Q^*_i}$ 
where $Q^*_i = q_i^*(q_i)$ is the optimal quantity to be produced by firm $i$ 
when there is no other firms to compete with.
\end{proof}

\bibliographystyle{plainnat}
\bibliography{cournot}

\medskip

\section{Example}
Here we write the set of equations of our example in more details.
\begin{description}
\item[Scenario 1]
The set of 
equations $\frac{\partial \pi_A}{\partial q_{A1}}=\frac{\partial \pi_B}{\partial q_{B1}}=0$ can be written as:
\begin{eqnarray*}
\frac{\partial \pi_A}{\partial q_{A1}}& = &1-3q_{A1}-q_{B1}=0 \\
\frac{\partial \pi_B}{\partial q_{B1}}& = &1-3q_{B1}-q_{A1}=0
\end{eqnarray*}

\item[Scenario 2]
The set of equations 
$\frac{\partial \pi_A}{\partial q_{A1}}=\frac{\partial \pi_A}{\partial q_{A2}}=\frac{\partial \pi_B}{\partial q_{B1}}=\frac{\partial \pi_B}{\partial q_{B2}}=0$ can be written as:
\begin{eqnarray*}
\frac{\partial \pi_A}{\partial q_{A1}} & = & 1-5q_{A1}-q_{A2}-2q_{B1} = 0 \\
\frac{\partial \pi_A}{\partial q_{A2}} & = & 1-5q_{A2}-q_{A1}-2q_{B2} = 0 \\
\frac{\partial \pi_B}{\partial q_{B1}} & = & 1-5q_{B1}-q_{B2}-2q_{A1} = 0 \\
\frac{\partial \pi_B}{\partial q_{B2}} & = & 1-5q_{B2}-q_{B1}-2q_{A2} = 0.
\end{eqnarray*}

\item[Scenario 3]
The set of equations 
$\frac{\partial \pi_A}{\partial q_{A1}}=\frac{\partial \pi_A}{\partial q_{A2}}=\frac{\partial \pi_B}{\partial q_{B2}}=0$ can be written as:
\begin{eqnarray*}
\frac{\partial \pi_A}{\partial q_{A1}} & = & 1-5q_{A1}-q_{A2} = 0 \\
\frac{\partial \pi_A}{\partial q_{A2}} & = & 1-5q_{A2}-q_{A1}-2q_{B2} = 0 \\
\frac{\partial \pi_B}{\partial q_{B2}} & = & 1-5q_{B2}-2q_{A2} = 0. 
\end{eqnarray*}

\end{description}

\section{Zhao and Han Convergence Theorem}
The following theorem states the performance guarantee of 
the algorithm proposed by \citeauthor{zhao1999two}~\citeyear{zhao1999two}.
\begin{theorem}[Zhao Han Convergence Theorem]\label{interior}
Let $F : \mathbb{R}^n \rightarrow \mathbb{R}^n$ be the function associated 
with a complementarity problem satisfying the two following conditions:
\begin{itemize}
\setlength{\itemindent}{.2in}
\item $\nabla F$ is a positive semidefinite matrix for a constant scalar.
\item {$F$ satisfies SLC; i.e., for some scalar $\lambda > 0$,
\begin{align*}
\Vert X[F(x+h)-F(x)-\nabla F(x) h]\Vert_\infty \le \lambda \vert h^T \nabla F(x) h \vert
\end{align*}
holds $\forall x > 0$ and $\forall h$ satisfying $\Vert X^{-1} h \Vert \le 1$}.
\end{itemize}
Then the algorithm converges in time 
$O\bigg(n\max(1,\lambda)~\log(\mu_0/\epsilon)\bigg)$ 
and outputs an approximate solution $(F(x^*),x^*)$ satisfying 
$(x^*)^T F(x^*)/n \le \epsilon$ where $\mu_0 = (x_0)^T F(x_0)/n$, 
and $(F(x_0),x_0)$ is the initial feasible point.
\end{theorem}

\section{Missing proofs}
\subsection{Missing proofs of Section \ref{general}}
\begin{proofof}{Lemma \ref{separation}}
Definition of function $F$ implies
\begin{align*}
\Vert X[F(q+h)-F(q)-\nabla F(q)h]\Vert_\infty = &
\Vert X[R(q+h)-R(q)-\nabla R(q)h] \\& + X[S(q+h)-S(q)-\nabla S(q)h] \Vert_\infty
\\\intertext{applying triangle inequality gives
}
\Vert X[F(q+h)-F(q)-\nabla F(q)h]\Vert_\infty\leq&\Vert X[R(q+h)-R(q)-\nabla R(q)h]\Vert_\infty \\
&+ \Vert X[S(q+h)-S(q)-\nabla S(q)h] \Vert_\infty
\end{align*}
Combining with assumptions of the lemma, we have the required inequality.
\end{proofof}


\begin{proofof}{Lemma \ref{SLCLHS_R}}
Before we proceed, we state the following theorem from analysis and Lemma \ref{lem:minor}.
\begin{theorem}\cite{boydbook}
  \label{thm:quad}
  Let $f:\mathbb{R}^n \mapsto \mathbb{R}$ be a continuously differentiable function with Lipschitz gradient, i.e., for some scaler $c > 0$,
  \begin{align*}
    \Vert \nabla f(x) -  \nabla f(y) \Vert \leq c \Vert x - y \Vert ~~\forall~x,y \in \mathbb{R}^n.
  \end{align*}
  Then, we have $\forall~x,y \in \mathbb{R}^n$,
  \begin{align}
    f(y) \leq f(x) + \nabla f(x)^T(x-y) + \frac{c}{2}\Vert y - x\Vert^2\label{eq:quad}
  \end{align}
\end{theorem}
\begin{lemma}
\label{lem:minor}
  For any vector $x \in \mathbb{R}^n$ and an arbitrary $S \subseteq [n]$, let $X = \sum_{i \in S}x_i$. Then we have $\sqrt{n}\Vert x \Vert \geq X$
\end{lemma}
\begin{proof}
  Let $Y = \sum_{i \in [n]}|x_i|$. Clearly, $|Y| \geq |X|$.
  \begin{align*}
    Y^2 = \sum_{i,j \in [n]}|x_ix_j| = \sum_{i < j}2|x_ix_j| + \Vert x \Vert^2 
  \end{align*}
  Since, $s^2 + t^2 \geq 2st~\forall~s,t \in \mathbb{R}$, we have
  \begin{align*}
    X^2 \leq Y^2 \leq \sum_{i < j}(x^2_i + x^2_j) + \Vert x \Vert^2 = n\Vert x \Vert^2 
  \end{align*}
\end{proof}

Now we are ready to prove Lemma \ref{SLCLHS_R}.
First note that $R({q}+h)-R({q})-\nabla R({q})h$
is an $E \times 1$ vector. Let $H_i = \sum_{j \in N_{\mathcal{M}}(i)}h_{ij}$. The element corresponding to edge $(i,j) \in \mathcal{E}$ 
in vector $R({q}+h)$ is $P_i(D_i+H_i)+P_i'(D_i+H_i)(q_{ij}+h_{ij})$. 
Similarly, the element corresponding to edge $(i,j)\in \mathcal{E}$ in $R({q})$ 
is $P_i(D_i)+P_i'(D_i)q_{ij}$ whereas the corresponding element in $\nabla R({q})h$ 
is $\sum_{k\in N_{\mathcal{M}}}{h_{ik} \frac{\partial r_{ij}}{\partial{q_{ik}}}} = -\sum_{k \in N_{\mathcal{M}}(i)}{h_{ik}(P_i'(D_i)+P_i''(D_i)q_{ij}))}+h_{ij}P_i'(D_i)$.
Therefore, the element corresponding to edge $(i,j) \in \mathcal{E}$ 
in vector $R({q}+h)-R({q})-\nabla R({q})h$ is:
\begin{align*}
-P_i(D_i+H_i)-&P_i'(D_i+H_i)(q_{ij}-h_{ij})-P_i(D_i)-P_i'(D_i)q_{ij}\\ 
&+\sum_{k \in N_{\mathcal{M}}(i)}{h_{ik}(P_i'(D_i)+P_i''(D_i)q_{ij})}+h_{ij}P_i'(D_i).
\end{align*}

Besides, $X$ is the diagonal matrix of size $E \times E$ with diagonal entries 
equal to elements of ${q}$ in the same order. 
Therefore, $X[R({q}+h)-R({q})-\nabla R({q})h]$ is 
an $E \times 1$ vector where the element corresponding to edge 
$(i,j) \in \mathcal{E}$ is $q_{ij}$ multiplied by the element corresponding 
to edge $(i,j)$ in vector $R({q}+h)-R({q})-\nabla R({q})h$:
\begin{align*}
&-q_{ij} \Bigg(P_i(D_i+H_i)+P_i'(D_i+H_i)(q_{ij}+h_{ij})-P_i(D_i)-P_i'(D_i)q_{ij}
\\
& \qquad-\sum_{k \in N_{\mathcal{M}}(i)}{h_{ik}(P_i'(D_i)+P_i''(D_i)q_{ij})}-h_{ij}P_i'(D_i) \Bigg)
\\
=&
-q_{ij} \Bigg(\left[P_i(D_i+H_i)-P_i(D_i)-H_i P_i'(D_i)   \right] 
\\
&\qquad+ \left[P_i'(D_i+H_i)-P_i'(D_i)-H_i P_i''(D_i)   \right] (q_{ij}+h_{ij}) + h_{ij} H_i P_i''(D_i)\Bigg) 
\\ \le& q_{ij} \Bigg(\vert P_i(D_i+H_i)-P_i(D_i)-H_i P_i'(D_i)   \vert 
\\
&\qquad+ \vert P_i'(D_i+H_i)-P_i'(D_i)-H_i P_i''(D_i)\vert \vert (q_{ij}+h_{ij})\vert + \vert h_{ij} H_i P_i''(D_i)\vert \Bigg).
\end{align*}

Let $P'$ and $P''$ be Lipschitz continuous functions with Lipschitz constants $2L_1$ and $2L_2$ respectively. To bound the last expression, we use Theorem \ref{thm:quad} and Lemma \ref{lem:minor}
\begin{align*}
\vert P_i(D_i+H_i)-P_i(D_i)-H_i P_i'(D_i)\vert
&\leq L_1H_i^2 \leq L_1E\Vert h \Vert^2\\
\vert P_i'(D_i+H_i)-P_i'(D_i)-H_i P_i''(D_i)\vert
&\leq L_2H_i^2 \leq L_2E\Vert h \Vert ^2\\
|h_{ij} H_i P_i''(D_i)| &\leq E\Vert h \Vert^2P_i''(D_i) 
\end{align*}
Then, from finiteness of derivatives, we have:
\begin{align*}
  |h_{ij} H_i P_i''(D_i)| \leq EM_2\Vert h \Vert^2
\end{align*}
Thus, the LHS is bound from above by:
\begin{align*} 
  q_{ij}E\Vert h \Vert^2(L_1 + L_2(q_{ij} + h_{ij}) + M_2)
\end{align*}
Let $Q$ be an upper bound on maximum profitable quantity for any producer in any market.
Then the LHS is bound above by $C\Vert h \Vert^2$, where: 
\begin{align} 
  \label{eq:rbound}
  C=QE(L_1 + 2QL_2 + M_2)
\end{align}
\end{proofof}

\begin{proofof}{Lemma \ref{SLCLHS_S}}
  Let $m_{ij} = \frac{\partial{c_{i}}}{\partial {q_{ij}}}$.
The element of vector $X(S(q+h) - S(q) - h\nabla S)$ corresponding to edge $(i,j)$ is given by:
\begin{align*}
  q_{ij}(m_{ij}(q+h)+m_{ij}(q)-h\nabla c_i(q))
\end{align*}
Let $2L_3$ be an upper bound of Lipschitz constants for derivates of $c_i$'s. Then, from Theorem \ref{thm:quad} and upper bound $Q$ on production quantities, we have:
\begin{align*}
  |q_{ij}(m_{ij}(q+h)+m_{ij}(q)-h\nabla c_i(q))| \leq QL_3\Vert h \Vert^2
\end{align*}
\end{proofof}

\begin{proofof}{Lemma \ref{lem:smono}}
  Let $S' = \sum_{i \in \mathcal{F}}c_i$. Then $S'$, being a sum of strongly convex functions,  is a strongly convex function. Also, $S = \nabla S'$. Thus, $h^T \nabla^2 S' h = h^T \nabla S h$ is bounded from below by $\alpha_c\Vert h \Vert^2, \forall h \in \mathbb{R}^n$ for some $\alpha_c > 0$ if the cost functions are strongly convex and $alpha_c = 0$ is cost functions are convex.
\end{proofof}

\begin{proofof}{Lemma \ref{SLC}}
  From lemmas \ref{SLCLHS_R}, \ref{SLCLHS_S} and \ref{separation}, RHS of SLC for $F$ is $O(E\Vert h \Vert^2)$. If cost functions are strongly convex or marginal revenue function is strongly monotone, then from Lemma \ref{lem:smono} and definition of strong monotonicity, the LHS of SLC for $F$ is $\Omega(\Vert h \Vert^2)$. Thus, $F$ satisfies SLC. We note that $F$ is a sum of two monotone functions and hence is monotone.
\end{proofof}

\begin{proofof}{Lemma \ref{PSD_R}}
Let $e_1$ be the index of the edge $(i,j)$ and $e_2$ be the index of edge $(l,k)$. 
The elements of $\nabla R$ are as follows.
\begin{align*}
\nabla R_{e_1 e_1} = \left\{
    \begin{array}{ll}
 \frac{\partial r_{ij}}{\partial q_{ij}} &= -2P_i'(D_i) -P_i''(D_i)q_{ij}  ~\mbox{if}~ e_1 = e_2	\\\nonumber
\frac{\partial r_{ij}}{\partial q_{ik}} &= -P_i'(D_i) - P_i''(D_i)q_{ij}~ \mbox{if}~ i = l, j \neq k\\\nonumber  
\frac{\partial r_{ij}}{\partial q_{lk}} &= 0 ~ \mbox{if}~ i \neq l, j \neq k.
    \end{array}
\right.
\end{align*}

We note that since price functions are functions only of the total production 
in their corresponding markets and not the individual quantities produced by firms, 
$\frac{\partial P_i'(D_i)}{\partial q_{ij}} = \frac{\partial P_i'(D_i)}{\partial q_{ik}}$. 
Therefore, we have replaced the partial derivatives by $P_i''(D_i)$. 

We must show ${x}^T \nabla R(D_i) {x}$ is nonnegative $\forall {x} \in \mathbb{R}^E$ and $\forall D_i \geq 0$.

\begin{align*} 
{x}^T(\nabla R(D_i)) {x} &= \sum_{(i,j) \in \mathcal{E}}{ \sum_{(k,l)\in \mathcal{E}} {x_{ij} x_{lk} \frac{\partial r_{ij}}{\partial x_{lk}} }} = \sum_{i \in \mathcal{M}}{\sum_{j,k\in N_{\mathcal{M}}(i)}{ x_{ij} x_{ik}\frac{\partial r_{ij}}{\partial x_{ik}}}}
 \\
&= \sum_{i\in \mathcal{M}}{\left(\sum_{j \in N_{\mathcal{M}}(i)} {x_{ij}^2 \left[-2P_i'(D_i)-P_i''(D_i)x_{ij}\right]} \right. }
\\
&\hspace{1.5cm}+ {\left.\sum_{j,k \in N_{\mathcal{M}}(i) , j\neq k} {x_{ij} x_{ik} \left[-P_i'(D_i)-P_i''(D_i)q_{ij}\right]}\right)} 
\\ 
&= -\sum_{i\in \mathcal{M}}{\left(\sum_{j \in N_{\mathcal{M}}(i)} {x_{ij}^2 P_i'(D_i)} + \sum_{j,k \in N_{\mathcal{M}}(i) } {x_{ij} x_{ik} (P_i'(D_i)+P_i''(D_i)q_{ij})}\right)} 
\\ 
 &=  -\sum_{i\in \mathcal{M}}{\left(P_i'(D_i)\hspace{-2mm}\sum_{j \in N_{\mathcal{M}}(i)}\hspace{-2mm} {x_{ij}^2} +  P_i'(D_i)\hspace{-2mm}\sum_{j,k \in N_{\mathcal{M}}(i) }\hspace{-2mm} {x_{ij} x_{ik}}  +  P_i''(D_i)\hspace{-2mm}\sum_{j,k \in N_{\mathcal{M}}(i) }\hspace{-2mm} {x_{ij}q_{ij}x_{ik}}\right)} \\
 &\geq -\sum_{i\in \mathcal{M}}{\left(P_i'(D_i)\hspace{-2mm}\sum_{j \in N_{\mathcal{M}}(i)}\hspace{-2mm} {x_{ij}^2} +  P_i'(D_i)\hspace{-2mm}\sum_{j,k \in N_{\mathcal{M}}(i)}\hspace{-2mm} {x_{ij} x_{ik}}  -  |P_i''(D_i)||{q}||{x}||\hspace{-2mm}\sum_{j\in N_{\mathcal{M}}(i)}\hspace{-2mm} {x_{ij}}|\right)} \\
 &\geq \sum_{i\in \mathcal{M}}{\left(-P_i'(D_i)|{x}|^2 - P_i'(D_i)\left(\sum_{j \in N_{\mathcal{M}}(i)} {\hspace{-3mm}x_{ij}}\right)^2 + |P_i''(D_i)| D_i|{x}||\hspace{-2mm}\sum_{j\in N_{\mathcal{M}}(i)} {\hspace{-2mm}x_{ij}}|\right)}
\end{align*}
Since $P_i$'s are decreasing functions, we have , $P_i'(D_i) \le 0, \quad\forall i \in \mathcal{M}$. Thus, over domain of $P_i$'s ($D_i \geq 0$), the above expression is non-negative if $|P''_i(D_i)|D_i \leq 2|P_i'(D_i)|$
Hence, ${x}^T(\nabla R(D_i)) {x} \ge 0$ equivalently $\nabla R(D_i)$ is positive semidefinite.
\end{proofof}

\subsection{Missing proofs of Section \ref{sec-1}}
\begin{proofof}{Lemma \ref{lem:bin:1}}
Again let $P'(Q) = P(Q+1)-P(Q)$ be the forward difference of the price function, and let $c'_i(q_i) = c_i(q_i+1) - c_i(q_i)$. 
From definition of profit function $\pi_i$ and $f_i$, we have $f_i(q_i,Q) = P(Q+1) + P'(Q)q_i - c'_i(q_i)$. Assume $Q$ is fixed. Suppose we have $q_i < \tilde{q_i}$. 
The marginal profit of firm at production quantity $q_i$ 
is $P(Q+1) + P'(Q)q_i - c'_i(q_i)$ whereas the marginal profit 
at production  quantity $\tilde{q_i}$ is 
$P(Q+1) + P'(Q)\tilde{q_i} - c'_i(\tilde{q_i})$. 
Thus, $P(Q+1) + P'(Q)q_i > P(Q+1) + P'(Q)\tilde{q_i}$ 
since $P'(Q)$ is negative (from concavity of $P$) and $q_i < \tilde{q_i}$.
As the discrete cost functions are convex, we have $c_i'(q_i) < c_i'(\tilde{q_i})$. 
This implies $f_i(q_i,Q) > f_i(\tilde{q_i},Q)$ when $q_i < \tilde{q_i}$. 
Thus, for a fixed $Q$, $f_i(q_i, Q)$ is a non-increasing function of $q_i$.
Similarly, we can see that $f_i(q_i , Q)$ is a non-increasing function of $Q$.
From definitions of $F_i$ and $G_i$, we have:
\begin{align}
  F_i(q_i+1, Q) - F_i(q_i, Q) &= f_i(q_i, Q) \label{eq:dif1}\\
  G_i(q_i+1, Q) - G_i(q_i, Q) &= F_i(q_i+1, Q-1) - F_i(q_i, Q-1) = f_i(q_i, Q-1) \label{eq:dif2}
\end{align}
For a fixed $Q$, Let $q_l$ be the minimum maximizer of $F_i(q_i, Q)$. Then $f_i(q_l-1,Q) > 0$. Let $q_u$ be the maximum maximizer of $G_i(q_i, Q)$. Because $f_i$ is non-increasing, we have $f_i(q_l-1, Q-1) \geq f_i(q_l-1, Q) > 0$. Thus, any number smaller than $q_l$ cannot be a maximizer of $G_i$ and we have $q_l \leq q_u$. Let $q \in \{q_l\ldots q_u\}$. Then, because $q \geq q_l$ we have $f_i(q,Q) \leq 0$ and from $q \leq q_u$, we have $f_i(q-1, Q-1) \geq 0$. Thus, $q$ is an equilibrium quantity when total production quantity is $Q$. If $q < q_l$, then $f_i(q, Q) > 0$ and if $q > q_u$ then $f_i(q-1, Q-1) > 0$. Thus $\{q_l\ldots q_u\}$ is the set of equilibrium quantities.
In Line~\ref{eq:alg1} of Algorithm \ref{alg:bin}, we are searching for the minimum maximizer of $F_i$ and in Line \ref{eq:alg2} we are searching for maximum maximizer of $G_i$. Binary search for these quantities is valid because first differences for both functions (equations \ref{eq:dif1} and \ref{eq:dif2}) are decreasing. 
\end{proofof}

\begin{proofof}{Lemma \ref{lem:bin:2}}
We use the following definition from submodular optimization in the lemma.
\begin{definition}
  Given lattices $(X_1, \geq)$ and $(X_2, \geq)$, $f:X_1 \times X_2 \mapsto \mathbb{R}$ is supermodular iff for any $x_1, y_1 \in X_1; x_2, y_2 \in X_2$ such that $x_1 \geq y_1$ and $x_2 \geq y_2$, the following holds:
  \begin{align*}
    f(x_1, y_2) - f(x_1, x_2) \geq f(y_1, y_2) - f(y_1, x_2) 
  \end{align*}
\end{definition}

  We have, $F_i(q_i, Q) = P(Q+1)q_i + \frac{P'(Q)}{2}(q_i-\frac{1}{2})^2 - c_i(q_i)$.
  Let $-Q_1 > -Q_2$. Let $q_i' > q_i$. Then, we have:
  \begin{align*}
    F_i(q_i, Q_1) - F_i(q_i, Q_2) = (P(Q_1+1) - P(Q_2+1))q_i +\frac{P'(Q_1) - P'(Q_2)}{2}(q_i-\frac{1}{2})^2\\
    F_i(q'_i, Q_1) - F_i(q'_i, Q_2) = (P(Q_1+1) - P(Q_2+1))q'_i +\frac{P'(Q_1) - P'(Q_2)}{2}(q'_i-\frac{1}{2})^2
  \end{align*}    
  Since $P$ and $P'$ are a decreasing functions, we have $P(Q_1) \geq P(Q_2)$ and $P'(Q_1) \geq P'(Q_2)$. From this and the fact that $q_i' > q_i$, we have:
  \begin{align*}
    F_i(q_i', Q_1) - F_i(q_i', Q_2) \geq F_i(q_i, Q_1) - F_i(q_i, Q_2) 
  \end{align*}    
  Therefore $F_i^-$ is a supermodular function. Since $G_i(q_i, Q) = F_i(q_i, Q-1)$, may a similar argument we can conclude that $G_i^-$ is supermodular.
\end{proofof}

\begin{proofof}{Lemma \ref{lem:bin:3}}
We need the following definition and Topkis' Monotonicity Theorem for proving the lemma.
\begin{definition}
  Given a lattice $(X, \geq)$, we define \emph{Strong Set Ordering} over $A,B \subseteq X$.  We say $A \geq_s B$ iff $\forall a \in A, \forall b \in B, \max\{a,b\} \in A \land \min\{a,b\} \in B$.
\end{definition}

We note that the strong set ordering induces a natural ordering on sets of consecutive integers. Let $I_1=\{l_1,\ldots ,u_1\}$. Let $I_2=\{l_2,\ldots ,u_2\}$. Then $I_1 \geq_s I_2$ iff $l_1 \geq l_2$ and $u_1 \geq u_2$.

\begin{theorem}[Topkis' Monotonicity Theorem\cite{topkis1978minimizing}]
  \label{thm:topkis}
  For any lattices $(X, \geq)$ and $(T, \geq)$, let $f: X \times T \mapsto \mathbb{R}$ be a supermodular function and let $x^*(t) = argmax_{x \in X}f(x, t)$. If $t \geq t'$, then $x^*(t) \geq_s x^*(t')$, i.e., $x^*(t)$ is non-decreasing in $t$.
\end{theorem}

We note that in the theorem above, strong set ordering is used over $x^*$ because $argmax$ returns a set of values from lattice $X$.

Now we are ready to prove Lemma \ref{lem:bin:3}.  From Lemma \ref{lem:bin:2}, $F_i^-(q_i, -Q)$ is a supermodular function. Thus, from Theorem \ref{thm:topkis}, $q_i^F(Q) = argmax F_i(q_i, Q)$ is a non-decreasing function of $-Q$, i.e., $q_i^F$ is a non-increasing function of $Q$. Thus $Q > Q' \implies I' \geq_s I$. As noted above, strong set ordering on a set of consecutive integers implies that $q^{'l}_i \geq q_i^l$ and $q^{'u}_i \geq q_i^u$.
\end{proofof}

\begin{proofof}{Lemma \ref{lem:bin:3a}}
From Lemma \ref{lem:bin:2}, $G_i^-(q_i, -Q)$ is a supermodular function. Thus, from Theorem \ref{thm:topkis}, $q_i^G(Q) = argmax F_i(q_i, Q)$ is a non-decreasing function of $-Q$, i.e., $q_i^G$ is a non-increasing function of $Q$. Thus $Q > Q' \implies I' \geq_s I$. As noted above, strong set ordering on a set of consecutive integers implies that $q^{'l}_i \geq q_i^l$ and $q^{'u}_i \geq q_i^u$.
\end{proofof}
\begin{proofof}{Lemma \ref{lem:bin:4}}
  Assume for contradiction that such an equilibrium exists for total quantity $Q' > Q$. From Lemma \ref{lem:bin:3a}, we have $q_i^G(Q) \geq_s q_i^G(Q') = \{q^{'l}_i \ldots q^{'u}_i\}$. Thus, we have $q_i^{u} \geq q^{'u}_i$. Since $Q'$ is an equilibrium quantity, $\sum_{i=1}^n q_i^{'u} \geq Q'$. Thus, we have $Q' < Q$ and this is a contradiction.
\end{proofof}

\begin{proofof}{Lemma \ref{lem:bin:5}}
  Assume for contradiction that such an equilibrium exists for total quantity $Q' < Q$. From Lemma \ref{lem:bin:3}, we have $q_i^F(Q) \leq_s q_i^F(Q') = \{q^{'l}_i \ldots q^{'u}_i\}$. Thus, we have $q_i^{l} \leq q^{'l}_i$. Since $Q'$ is an equilibrium quantity, $\sum_{i=1}^n q_i^{'l} leq Q'$. Thus, we have $Q' > Q$ and this is a contradiction.
\end{proofof}

\end{document}